%% file: Mao21Main.tex
\documentclass[final,12pt]{colt2021} 

\title[Random Graph Matching with Improved Noise Robustness]{Random Graph Matching with Improved Noise Robustness}
\usepackage{times}

\coltauthor{%
 \Name{Cheng Mao} \Email{cheng.mao@math.gatech.edu}\\
 \addr School of Mathematics, Georgia Institute of Technology
 \AND
 \Name{Mark Rudelson} \Email{rudelson@umich.edu}\\
 \addr Department of Mathematics, University of Michigan
 \AND
 \Name{Konstantin Tikhomirov} \Email{ktikhomirov6@gatech.edu}\\
 \addr School of Mathematics, Georgia Institute of Technology%
}

\usepackage{amsmath,amssymb,latexsym}
\usepackage{enumitem,bbm}

\newcommand{\N}{\mathbb{N}}

\newcommand{\R}{\mathbb{R}}

\newcommand{\Exp}{\mathbb{E}}

\newcommand{\Event}{\mathcal{E}}
\def\Prob{{\mathbb P}}

\newcommand{\sign}{\mathbf{Sign}}

\def\neigh{{\mathcal N}}

\newcommand{\1}{\mathbbm{1}}
\newcommand{\E}{\mathbb{E}}
\newcommand{\p}{\mathbb{P}}
\newcommand{\cE}{\mathcal{E}}
\newcommand{\cI}{\mathcal{I}}
\newcommand{\cJ}{\mathcal{J}}
\newcommand{\Bin}{\mathrm{Binomial}}
\newcommand{\Ber}{\mathrm{Bernoulli}}

\newcommand{\SA}{\mathbf{A}}
\newcommand{\SB}{\mathbf{B}}
\newcommand{\SC}{\mathbf{C}}

\begin{document}

\maketitle

\begin{abstract}%
  Graph matching, also known as network alignment, refers to finding a bijection between the vertex sets of two given graphs so as to maximally align their edges. This fundamental computational problem arises frequently in multiple fields such as computer vision and biology. Recently, there has been a plethora of work studying efficient algorithms for graph matching under probabilistic models. In this work, we propose a new algorithm for graph matching: Our algorithm associates each vertex with a signature vector using a multistage procedure and then matches a pair of vertices from the two graphs if their signature vectors are close to each other. We show that, for two Erd\H{o}s--R\'enyi graphs with edge correlation $1-\alpha$, our algorithm recovers the underlying matching exactly with high probability when $\alpha \le 1 / (\log \log n)^C$, where $n$ is the number of vertices in each graph and $C$ denotes a positive universal constant. This improves the condition $\alpha \le 1 / (\log n)^C$ achieved in previous work.%
\end{abstract}


\input{introduction}

\input{algorithm}

\input{proof}
\acks{We thank the anonymous reviewers for their helpful comments. C.M.\ was partially supported by the NSF grant DMS-2053333. M.R.\ was partially supported by the NSF grant DMS-2054408. K.T.\ was partially supported by the Sloan Research Fellowship. A part of the work was done when K.T.\ was visiting the University of Michigan in February 2020.}

\bibliography{matching}

\appendix

\input{appendix}

\end{document}

%% file: introduction.tex

\section{Introduction}
\label{sec:intro}

The problem of \emph{graph matching} or \emph{network alignment} consists in finding a bijection between the vertex sets of two given graphs in order to maximally align their edges. 
As a practical problem, graph matching has been studied in pattern recognition for decades \citep{conte2004thirty,emmert2016fifty} and finds applications in many other areas including network security \citep{narayanan2008robust,narayanan2009anonymizing} and computational biology \citep{singh2008global,kazemi2016proper}. 
Mathematically, graph matching can be formulated as a \emph{quadratic assignment} problem
$ \max_{\Pi} \langle A, \Pi B \Pi^\top \rangle , $
where $A$ and $B$ denote the adjacency matrices of the two given graphs respectively, and $\Pi$ is maximized over the set of permutation matrices. 
However, the quadratic assignment problem is known to be NP-hard to solve or approximate in the worst case \citep{Pardalos94thequadratic,burkard1998quadratic,makarychev2010maximum}. 
Even when the two given graphs are isomorphic, in which case graph matching reduces to the \emph{graph isomorphism} problem, the computational complexity is not settled \citep{babai2016graph}.

Fortunately, real-world networks can often be represented by average-case models which circumvent worst-case computational hardness. 
As a result, there has been a plethora of work in the literature studying \emph{random} graph matching. 
In the noiseless setting, a graph isomorphism between Erd\H{o}s--R\'enyi graphs can be found in linear time with high probability in the information-theoretically possible regime
\citep{babai1980random,bollobas1982distinguishing,czajka2008improved}. 
More recently, a model for correlated Erd\H{o}s--R\'enyi graphs has been proposed for graph matching by \cite{pedarsani2011privacy}. 
Since then, there has been a surge of interest in studying related models and algorithms in the literature \citep{yartseva2013performance,lyzinski2014seeded,kazemi2015growing,feizi2016spectral,cullina2016improved,shirani2017seeded,DMWX18,barak2019nearly,bozorg2019seedless,cullina2019partial,dai2019analysis,mossel2020seeded,fan2020spectral,ganassali20a,hall2020partial,racz2020correlated}. 
Our work falls into this category. 

\subsection{Correlated Erd\H{o}s--R\'enyi graph model}
\label{sec:model}

Let us formally state the correlated Erd\H{o}s--R\'enyi graph model \citep{pedarsani2011privacy} before further discussion. 
Let $n$ denote the number of vertices of each graph in consideration. 
Fix $p\in (0,1)$, $\alpha \in [0, 1-p]$ and a positive integer $n$. Let $G_0$ be a $G(n, \frac{p}{1-\alpha})$ Erd\H{o}s--R\'enyi graph, which is called the \emph{parent} graph.
From the graph $G_0$, a subgraph $G$ is obtained by removing every edge of $G_0$ independently with probability $\alpha$. 
Next, another subgraph $G'$ of $G_0$ is obtained in the same fashion, independent from $G$ conditionally on $G_0$. 
Then $G$ and $G'$ are marginally both $G(n, p)$,
and for every pair of distinct indices $i$ and $j$ in $\{1, 2, \dots, n\}$,
$$
\Prob\big\{\mbox{$i$ is adjacent to $j$ in $G$}\,|\,\mbox{$i$ is adjacent to $j$ in $G'$}\big\}
=1-\alpha.
$$
Therefore, $\alpha$ is the noise level in the model. 
Fix an unknown permutation $\pi : [n] \to [n]$. Let $G^{\pi}$ denote the graph obtained from permuting the vertices of $G$ by $\pi$. In other words, $i$ is adjacent to $j$ in $G$ if and only if $\pi(i)$ is adjacent to $\pi(j)$ in $G^{\pi}$. 
The latent permutation $\pi$ represents the unknown matching between the vertices of the two graphs. 
Observing the graphs $G^{\pi}$ and $G'$, we aim to recover the permutation $\pi$. 

\subsection{Previous results and our improvement}

To ease the discussion, we use the standard asymptotic notation $O(\cdot)$, $o(\cdot)$, $\Omega(\cdot)$, $\Theta(\cdot)$, \dots with the understanding that $n$ is growing. We use $\tilde O(\cdot)$ to hide a polylogarithmic factor in $n$. Moreover, we use $C, C', c, c'$, possibly with subscripts, to denote universal positive constants that may change at each appearance. 

For the above model, \cite{cullina2016improved} studied the information-theoretic threshold for the exact recovery of $\pi$. 
In particular, it was shown that $\pi$ can be recovered with high probability if $np(1-\alpha) = \Omega( \log n )$, but no polynomial-time algorithm is known to achieve such a condition. 
Under the condition $np \ge n^{o(1)}$ and $\alpha \le 1 - (\log n)^{-o(1)}$ where $o(1)$ denotes a sub-constant quantity, \cite{barak2019nearly} provided an algorithm with time complexity $n^{O(\log n)}$ that recovers $\pi$ with high probability. 
Furthermore, under the condition $np \ge (\log n)^{C_1}$ and $\alpha \le (\log n)^{-C_2}$, \cite{DMWX18} proposed a polynomial-time algorithm based on degree profiles, and \cite{FMWX19a,FMWX19b} introduced a spectral method, both of which achieved exact recovery of $\pi$ with high probability. 
In addition, for sparse graphs with $C_1 \log n \le np \le e^{(\log \log n)^{C_1}}$, the method proposed by  \cite{DMWX18} recovers $\pi$ in the regime $\alpha \le (\log \log n)^{-C_2}$. 

Note that, in the regime where the average degree satisfies $np \ge (\log n)^{C_1}$, the computationally inefficient algorithms can tolerate a noise level $\alpha$ close to $1$, while the existing polynomial-time algorithms succeed only if $\alpha \le (\log n)^{-C_2}$. 
This prompts us to ask whether this noise condition can be improved. 
In this work, we propose a new efficient algorithm that recovers $\pi$ with high probability under the milder condition $\alpha \le (\log \log n)^{-C}$ for a constant $C>0$, which in particular extends the condition required by \cite{DMWX18} to all sparsity levels of the graphs. 
Moreover, our algorithm has time complexity $\tilde O(n^2)$ which is faster than the aforementioned existing methods. 

The high-level idea of our algorithm is to, in three stage, construct a matching between finer and finer \emph{partitions} of the vertex sets of the two graphs. 
At each stage of the algorithm, we obtain a partition $Q_1 \cup \cdots \cup Q_k$ of the vertex set of $G^\pi$ and a partition $Q'_1 \cup \cdots \cup Q'_k$ of the vertex set of $G'$, in a way that $Q_i$ is ``matched'' to $Q'_i$ for each $i = 1, \dots, k$. 
We have $k = \Theta(\log \log n)$, $k = \Theta(\log n)$, and $k = n$ at the three stages of the algorithm respectively. 
In other words, the matching obtained at the beginning is very coarse, while at the end we have a full matching between individual vertices of the two graphs. The algorithm is motivated and described in detail in Section~\ref{sec:algo-res}. 

\subsection{Other related work and theoretical significance}
There are also a few variants of the graph matching problem considered here. 
For example, in the regime where the average degree is $O(1)$, \cite{ganassali20a,hall2020partial} studied partial recovery of the matching $\pi$. 
Moreover, in the setting where we are given a few correctly matched pairs of vertices known as seeds, the problem is known as seeded graph matching and has been studied as well \citep{kazemi2015growing,mossel2020seeded,yu2020graph}. 
There have also been recent studies on the associated testing problem for correlations between unlabeled random graphs \citep{barak2019nearly,wu2020testing}. 
Finally, other correlated random graph models have been considered, such as a general Wigner model \citep{FMWX19b} and preferential attachment models \citep{korula2014efficient,racz2020correlated}. 
In many of these models, a constant noise level $\alpha$ appears to be a major bottleneck for computationally efficient algorithms. 
However, for exact recovery in the correlated Erd\H{o}s--R\'enyi graph model that we consider, neither an upper bound nor a lower bound of constant order is known. Our work makes one step towards understanding the threshold of the noise level. 


Furthermore, in recent years, there have been intensive efforts in studying average-case matrix models with a latent permutation as the planted signal. 
Besides graph matching, other examples of such models include noisy sorting \citep{braverman2008noisy,mao2018minimax}, 
random assignment \citep{aldous2001zeta,moharrami2019planted}, and hidden nearest neighbor graphs \citep{cai2017detection,ding2020consistent}. 
Unlike models having a planted low-rank signal (such as community detection models \citep{abbe2017community} and spiked random matrix models \citep{perry2018optimality}), models with a planted permutation are far less well-understood, as the combinatorial nature of the latent permutation brings significant algorithmic challenges. 
Graph matching, being one of the simplest and most generic models in this class, therefore deserves further study. 
Our work shrinks the gap between the information-theoretic threshold and guarantees for polynomial-time algorithms, contributing to the understanding of average-case matrix models with planted permutations.

\paragraph{Notation}
For any positive integer $k$, let $[k]$ be the set of integers $\{1,2,\dots,k\}$. Let $\N_+$ denote the set of positive integers and $\N_0$ the set of nonnegative integers. 
Let $\land$ and $\lor$ denote the $\min$ and the $\max$ operator for two real numbers, respectively.
For two positive functions $f$ and $g$ depending on $n$, we will write $f\ll g$ if $f=o_n(g)$ when $n\to\infty$.

For a graph $G$ with vertex set $[n]$ and $i \in [n]$, let $\deg_G(i)$ denote the degree of $i$ in $G$. For a subset $S \subset [n]$, let $\neigh_G(i; S)$ denote the set of neighbors of $i$ within $S$. Let $G(S)$ denote the subgraph of $G$ induced by $S$. 
For $n \in \N_0$ and $p\in(0,1)$, we denote by $F_{n,p}$ the cumulative distribution function of $\Bin(n,p)$. 
Let $\sign : \R \to \{-1,1\}$ denote the sign function defined by $\sign(x) = -1$ if $x<0$ and $\sign(x) = 1$ if $x\ge 0$.

%% file: algorithm.tex

\section{Algorithm and result}
\label{sec:algo-res}

We first motivate and sketch the main steps of a simplified version of our algorithm in Section~\ref{sec:algo-simple}. Along the way, we explain at a high level why the algorithm attains improved noise robustness. We describe our algorithm in full detail in Section~\ref{sec:algo} and state the theoretical guarantee for exact recovery of the matching in Section~\ref{sec:guarantee}. 

\subsection{Sketch of a simplified algorithm and heuristics}
\label{sec:algo-simple}

Recall that we are given graphs $G^\pi$ and $G'$ according to the correlated Erd\H{o}s--R\'enyi graph model defined in Section~\ref{sec:model}, where $G^\pi, G' \sim G(n, p)$ marginally and the noise level $\alpha$ is assumed to be in $[0, (\log \log n)^{-C}]$ for a constant $C > 0$. 
The (simplified) algorithm consists of three stages. 
%
%
%

\paragraph{First-generation partitions}
Consider vertex $i$ of $G^\pi$ and vertex $j$ of $G'$ such that $i = \pi(j)$, that is, $i$ and $j$ should be matched. Because of the correlation between the two graphs, the degrees of $i$ and $j$ are correlated. 
However, when the noise level $\alpha$ is as large as $(\log \log n)^{-C}$, it is impossible to match the vertices using their degrees directly, because the stochasticity in degrees is too large. 


Instead of matching vertices of the two graphs directly, we will match \emph{partitions of the vertex sets} as a first step. More specifically, for a positive integer $m = \Theta( \log \log n )$, we partition the vertex set $[n]$ of $G^\pi$ into subsets $Q_1, \dots, Q_m$ by defining
$$
Q_\ell := \big\{ i \in [n] :\;F_{n,p} ( \deg_{G^{\pi}} (i) ) \in \big((\ell-1)/m, \ell/m\big]\big\} .
$$
In other words, vertices of $G^\pi$ with similar degrees are grouped together, and the precise cutoffs are determined by the $m$-quantiles of the $\Bin(n,p)$ distribution. 
Similarly, we partition the vertex set $[n]$ of $G'$ into subsets $Q'_1, \dots, Q'_m$. 
We refer to $\{Q_1, \dots, Q_m\}$ and $\{Q'_1, \dots, Q'_m\}$ as the \emph{first-generation partitions}. 
Using the correlation between the two graphs, we can show that 
$Q_\ell \approx \pi(Q'_\ell)$ for each $\ell \in [m]$ with high probability. 
That is, the set $Q_\ell$ of vertices of $G^\pi$ is ``matched'' to the set $Q'_\ell$ of vertices of $G'$. 

\paragraph{Second-generation partitions}
The next step is to use the matching between $m$ pairs of sets of vertices to obtain a finer matching between $2^m$ pairs of sets of vertices as follows. 
Namely, we partition the vertex set $[n]$ of $G^\pi$ into subsets $R_s$, $s \in \{-1, 1\}^m$, by defining  
$$
R_s := \big\{i\in [n]:\; \sign \big( |\neigh_{G^{\pi}}(i; Q_\ell)|- p |Q_\ell| \big) = s_\ell \mbox{ for all $\ell \in [m]$}\big\} . 
$$
Similarly, we partition the vertex set $[n]$ of $G'$ into subsets $R'_s$ for $s \in \{-1, 1\}^m$. 
We refer to $\{R_s : s \in \{-1, 1\}^m\}$ and $\{R'_s : s \in \{-1, 1\}^m\}$ as the second-generation partitions. 

If $(i, j)$ is a pair of vertices that should be matched, that is, $i = \pi(j)$, then the number of neighbors of $i$ in $Q_\ell$ in the graph $G^\pi$ is correlated with the number of neighbors of $j$ in $Q'_\ell$ in the graph $G'$. 
As a result, it is more likely that $i$ and $j$ belong to sets $R_s$ and $R'_s$ respectively in view of the above definitions. 
Therefore, the set $R_s$ of vertices of $G^\pi$ is matched to the set $R'_s$ of vertices of $G'$ in the sense that $R_s$ and $\pi(R'_s)$ have a significant overlap.

\paragraph{Vertex signatures}
Given the matched $2^m$ pairs $(R_s, R'_s)$ of sets of vertices, we are ready to define the final matching between vertices of the two graphs. 
The method is very similar to how the second-generation sets were obtained from the first. 
To be more precise, for each vertex $i$ of the graph $G^\pi$, we define a vector $f(i) \in \{-1,1\}^{2^m}$ by
$$
f(i)_s := \sign \big( |\neigh_{G^{\pi}}(i; R_s)|-p |R_s| \big) . 
$$
Similarly, we define such a vector $f'(j) \in \{-1,1\}^{\omega}$ for each vertex $j$ of $G'$. 
We refer to $f(i)$ and $f'(j)$ as the \emph{signatures} of $i$ and $j$ in the two graphs respectively. 

If $i = \pi(j)$, in view of the overlap between $R_s$ and $\pi(R'_s)$, we may show that the signs $f(i)_s$ and $f'(j)_s$ are correlated. 
Using this correlation for all $s \in \{-1,1\}^m$, we can get that, with high probability, the vectors $f(i)$ and $f'(j)$ are sufficiently ``close'' to each other if and only if $i = \pi(j)$. 
The failure probability turns out to be roughly $\exp( - 2^m )$, which is polynomially small in $n$ given that $m = \Theta( \log \log n )$. Consequently, all pairs of vertices of the two graphs can be correctly matched based on their signatures.

\bigskip

In summary, at a noise level $\alpha = (\log \log n)^{-C}$, we can first obtain a matching between $m = \Theta(\log \log n)$ pairs of first-generation sets of vertices in the two graphs. From there, we then get a refined matching between $2^m$ pairs of second-generation sets of vertices. Finally, another refinement yields the full matching between all $n$ pairs of vertices of the two graphs, where crucially $n \ll 2^{2^m}$. 

To reiterate the high-level heuristics, let us instead consider the case of a lower noise level $\alpha = (\log n)^{-C}$ for a sufficiently large constant $C>0$. 
In this case, we may instead take $m = \Theta(\log n)$ and still obtain a matching between $m$ pairs of first-generation sets of vertices. 
If, for example, $m \ge 2 \log_2 n$, then we already have $n \ll 2^m$. Hence, the second step of the algorithm can be skipped, and the signatures of vertices can be defined directly based on the first-generation sets to yield the full matching. 

\bigskip
The above outline of our algorithm omitted a few intricacies. 
Most importantly, the multiple stages of the algorithm result in probabilistic dependencies that are difficult to understand. 
To mitigate this issue, we will partition the vertex set $[n]$ of each graph into three groups, and run the three steps of the algorithm on these three groups of vertices respectively. 
This way, in particular, we are able to condition on a realization of the first-generation sets
to study the structure of the second generation, and then condition on both the first and the second generation to study the vertex signatures. 
Another issue is that, even with the above modification, the entries of a signature $f(i)$ are not independent, and we need an extra ``sparsification'' step to weaken the dependency across its entries. 
More specifically, this sparsification step entails choosing a uniform random subset $\cI \subset \{-1,1\}^{m}$ of suitable cardinality and using $\{R_s, R'_s : s \in \cI\}$ to construct vertex signatures, rather than employing all second-generation sets. 
The full algorithm is described in the following subsection.

\subsection{The full algorithm}
\label{sec:algo}

We propose the following algorithm, where $\beta$, $m$ and $\omega$ are parameters to be specified in Theorem~\ref{thm:main}.

\begin{enumerate}[leftmargin=*,label={\arabic*.}]
\item 
Partition the vertex set $[n]$ of the first graph $G^{\pi}$ uniformly at random into sets $\SA$, $\SB$ and $\SC$ of cardinalities $|\SA|= n - \beta n$ and $|\SB| = |\SC| = 0.5 \beta n$. 
Similarly, and independently from $(\SA,\SB,\SC)$, partition the vertex set $[n]$ of the second graph $G'$ into $\SA'$, $\SB'$ and $\SC'$ of the same corresponding cardinalities. 

\item 
Partition the set $\SA$ into subsets $Q_1, \dots, Q_m$ according to the degrees of vertices in the induced graph $G^{\pi}(\SA)$. More precisely, for each $\ell \in [m]$, define 
\begin{align}
Q_\ell:=\big\{ i \in \SA :\;F_{|\SA|,p} ( \deg_{G^{\pi}(\SA)} (i) ) \in \big((\ell-1)/m, \ell/m\big]\big\}  .
\label{eq:first-gen}
\end{align}
Similarly, partition the set $\SA'$ into subsets $Q'_1,\dots,Q'_m$ according to $G'$.

\item 
Partition the set $\SB$ into $2^m$ subsets $R_s$ for $s \in \{-1,1\}^m$ according to the number of neighbors of $i \in \SB$ within $Q_1, \dots, Q_m$ in the graph $G^{\pi}$. More precisely, for each $s\in\{-1,1\}^m$, let
\begin{align}
R_s := \big\{i\in \SB:\; \sign \big( |\neigh_{G^{\pi}}(i; Q_\ell)|- p |Q_\ell| \big) = s_\ell \mbox{ for all $\ell \in [m]$}\big\} . 
\label{eq:sec-gen}
\end{align}
Similarly, partition the set $\SB'$ into $2^m$ subsets $R'_s$ where $s \in \{-1,1\}^m$ according to $G'$. 

\item 
Choose a uniform random subset $\cI \subset \{-1,1\}^{m}$ of cardinality $\omega$. 
For each vertex $i\in \SC$, define a vector $f(i) \in \{-1,1\}^{\omega}$ by
\begin{align}
f(i)_s := \sign \big( |\neigh_{G^{\pi}}(i; R_s)|-p |R_s| \big) ,\quad s\in\cI.
\label{eq:sig}
\end{align}
Similarly, define a vector $f'(i) \in \{-1,1\}^{\omega}$ for each $i \in \SC'$ according to $G'$.

\item 
We say that two vertices $i\in \SC$ in $G^{\pi}$ and $j\in \SC'$ in $G'$ are \emph{potentially} matched if
\begin{align}
\sum_{s \in \cI} \1 \big\{ f(i)_s = f'(j)_s \big\} > \frac{\omega}{2} \Big( 1 + \frac{\beta}{\log\log n} \Big)
\label{eq:decide}
\end{align}
and are \emph{potentially} not matched otherwise. 

\item
Repeat Steps 1--5 $(\log \log n)^2/\beta^4$ times. For two vertices $i \in [n]$ in $G^{\pi}$ and $j \in [n]$ in $G'$, if \emph{every time}
that $\SC$ contains $i$ and $\SC'$ contains $j$, the two vertices $i$ and $j$ are potentially matched, then we match $i$ and $j$, that is, we define
$\hat \pi(j) := i$.
If a well-defined permutation $\hat \pi$ is obtained from considering all pairs of vertices $(i, j) \in [n]^2$, then return $\hat \pi$ as the estimator of $\pi$; otherwise, return ``error".
\end{enumerate}

Again, we refer to the class of sets of vertices $\{Q_1, \dots, Q_m\}$ and the counterpart $\{Q'_1, \dots, Q'_m\}$ as the first generation of sets. Similarly, we refer to $\{R_s : s \in \{-1,1\}^m \}$ and $\{R'_s : s \in \{-1,1\}^m \}$ as the second generation of sets. 
In addition, we refer to the vector $f(i)$ as the signature of $i$ in $G^{\pi}$ and
$f'(i)$ as the signature of $i$ in $G'$.

\subsection{Theoretical guarantee}
\label{sec:guarantee}

It is not difficult to see that the running time of Steps~1 -- 4 is at the order of the number of edges in the graph, while Step~5 has time complexity $O(n^2)$. Therefore, the overall time complexity of the algorithm is $\tilde O(n^2)$. 
Our main result is the following guarantee for the algorithm. 
\begin{theorem}
\label{thm:main}
For $\delta \in(0, 0.1)$, there exists a universal constant $n_0 = n_0(\delta) > 0$ such that the following holds. 
Fix $n \ge n_0$, $p \in (0, 1/2]$ and $\alpha \in [0, 1-p]$ such that 
\begin{align}
np \ge (\log n)^{13} , \quad
\alpha \le (\log \log n)^{-6-\delta} . 
\label{eq:conditions}
\end{align}
For a fixed unknown permutation $\pi : [n] \to [n]$, let $G^\pi$ and $G'$ be given by the correlated Erd\H{o}s--R\'enyi graph model with parameters $n, p$ and $\alpha$ defined in Section~\ref{sec:model}. 
Let us define
\begin{align}
\beta &:= \min \big\{ \beta' \ge (\log \log n)^{-6-\delta} : 0.5 \beta' n \text{ is an integer} \big\} , \label{eq:def-beta} \\
m &:= \lfloor 6 \log_2 \log n \rfloor , \label{eq:def-m} \\
\omega &:= \lfloor (\log n)^{(1 + 2 \delta)} \rfloor . \label{eq:def-omega} 
\end{align}
Then the algorithm described in Section~\ref{sec:algo} with parameters $\beta$, $m$ and $\omega$ has time complexity $\tilde O(n^2)$ and returns the exact matching $\hat \pi = \pi$ with probability at least $1 - n^{-(\log n)^{\delta/4}}.$
\end{theorem}

A few questions remain open: 
\begin{itemize}
\item
The assumption on the average degree $np > (\log n)^{13}$ can potentially be improved. The current bottleneck is \eqref{eq: aux 20920498}, where we study the correlation between the second-generation partitions: In short, if the average degree is not sufficiently large, then some of the second-generation sets can be too small with a nontrivial probability, which becomes an issue in the following steps of the algorithm. 

\item
If the parameter $p$ unknown, our algorithm can be run with an empirical estimate of $p$. For example, if we define $\hat p$ to be the total number of edges in a graph divided by $\binom{n}{2}$, then the error $\hat p - p$ is of order $\sqrt{p}/n$, which is much smaller than $p$. We can potentially show that the modified algorithm with parameter $\hat p$ works under a comparable set of conditions. However, this intricacy is left out in the current work.

\item
It would be interesting to know whether the same guarantee can be achieved if we simplify the algorithm by omitting the step of splitting the vertex set into three parts and the step of sparsification. 

\item
Our algorithm runs in three stages so as to refine the estimated matching between partitions of the vertex sets. 
If we consider more generations of partitions in an iterative fashion, can some multistage algorithm achieve exact recovery of the underlying matching for a noise level $\alpha = (\log \cdots \log n)^{-C}$ or even better? 

\item
Finally, as discussed in the introduction, a constant noise level appears to be a major bottleneck for computationally efficient algorithms in random graph matching. 
While our work makes one step towards a constant $\alpha$, the problem deserves further investigation. 
\end{itemize}

%% file: proof.tex

\section{Proofs}
\label{sec:proof}

In this section, we present the key steps of the proof of Theorem~\ref{thm:main}, with additional details left to Appendix~\ref{sec:add-pf}. 
The analysis of the algorithm proceeds as follows. 
In Section~\ref{sec:first-gen}, we establish that the first-generation sets have size roughly $n/m$ each, and that the corresponding
sets $Q_\ell$ and $\pi(Q'_\ell)$ have large intersections with high probability. 
In Section~\ref{sec:sec-gen}, we show that the intersection between the second-generation sets $R_s$ and $\pi(R'_s)$ is sufficiently large for most indices $s\in\{-1,1\}^m$ with high probability. 
In Section~\ref{sec:sparse}, we perform sparsification to extract collections of second-generation sets
$\{R_s : s\in\cI\}$ and $\{R_s' : s\in\cI\}$ which have negligible intersections $R_s
\cap \pi(R_{\tilde s}')$ for distinct $s,\tilde s\in\cI \subset \{-1,1\}^m$.
In Section~\ref{sec:ver-sig}, we compare the vertex signatures of the two graphs. 
In Section~\ref{1874601847608}, we put together the pieces and prove Theorem~\ref{thm:main}.

To simplify the notation, we assume without loss of generality that the latent matching $\pi$ is the identity and thus $G = G^{\pi}$ throughout the proof. 

\subsection{First-generation partitions}
\label{sec:first-gen}

The subsets $\SA$ and $\SA'$ generated at Step~1 of the algorithm have cardinalities very close to $n$, so the sets $\SA$ and $\SA'$ have a large intersection. 
This, together with the correlation between the edges of $G$ and $G'$, implies that the degrees of corresponding
vertices of $G$ in $\SA$ and of $G'$ in $\SA'$ should be close to each other with high probability. 
Therefore, the corresponding elements of the
partitions of $\SA$ into $Q_1, \dots, Q_m$
and of $\SA'$ into $Q_1', \dots, Q_m'$ have large intersections with high probability. 
This is formally stated in the following lemma. 

\begin{lemma}[Correlation of first-generation sets]
\label{lem:sym-diff-1}
For any $\delta>0$ there is $n_0(\delta) \in \N_+$ with the following property. Let $n\geq n_0$, and assume \eqref{eq:conditions}, \eqref{eq:def-beta} and \eqref{eq:def-m}.
With probability at least $1 - 5 n^3 \exp(-\beta^3 \sqrt{pn} )$, the first generation of sets defined by \eqref{eq:first-gen} satisfies
\begin{align}
|Q_\ell \cap Q'_\ell| \ge n/m - C n \sqrt{\beta \log(1/\beta)} ,
\quad
|Q_\ell \triangle Q'_\ell| \le C n \sqrt{\beta \log(1/\beta)}, 
\label{eq:cor-1}
\end{align}
for all $\ell \in [m]$, where $C>0$ is a universal constant, and $\triangle$ denotes the symmetric difference. 
\end{lemma}

By our choice of $m$ and $\beta$, the lemma guarantees that $|Q_\ell \triangle Q'_\ell| \ll |Q_\ell \cap Q'_\ell|$ for all $\ell \in [m]$ with high probability. In other words, $Q_\ell$ and $Q'_\ell$ almost coincide, so the first-generation partitions are well matched.  
The proof of Lemma~\ref{lem:sym-diff-1} can be found in Section~\ref{subs: fg corr} and is based on preliminary results established in Section~\ref{subs: er aux lem}.

\subsection{Second-generation partitions}
\label{sec:sec-gen}

We now turn to the second-generation partitions.  
Recall that $\SB$ and $\SB'$ are independent random sets of cardinality $0.5 \beta n$ each. 
We first show that with high probability, the cardinality of each second-generation set $R_s$ or $R'_s$ is of order $\beta n/2^m$, and the intersection $R_s \cap \SB'$ has cardinality of order $\beta^2 n/2^m$. 
Furthermore, using the correlation between the two graphs, we establish that the intersection $R_s \cap R'_s$ also has cardinality of order $\beta^2 n/2^m$ for most $s \in \{-1, 1\}^m$ with high probability. 
Therefore, $R_s \cap R'_s$ occupies a significant portion of $R_s \cap \SB'$ for most $s$, which is the basis for constructing correlated vertex signatures later. 
These estimates are made precise in the following two lemmas. 


\begin{lemma}[Cardinalities of second-generation sets]
\label{lem:card-2}
For any $\delta>0$ there is $n_0(\delta)$ with the following property. Let $n\geq n_0$, and assume \eqref{eq:conditions}, \eqref{eq:def-beta} and \eqref{eq:def-m}. The following statements hold: 
\begin{itemize}
\item 
We have
\begin{align}
\frac{\beta^2 n}8 \le |\SB \cap \SB'| \le \frac{ \beta^2 n}2 
\label{eq:b-size}
\end{align}
with probability at least $1- 2 \exp( - \beta^2 n /40 )$. 

\item 
Conditional on any realization of $\{\SA,\SA',G(\SA),G(\SA'),\SB,\SB'\}$ such that \eqref{eq:cor-1} and \eqref{eq:b-size} hold, 
we have
\begin{align}
\frac{\beta n}{ 2^{m+3} } \le |R_s|, |R'_s| \le \frac{\beta n}{ 2^{m-1} },\quad
\frac{\beta^2 n}{ 2^{m+5} } \le |R_s \cap \SB'| \le \frac{\beta^2 n}{ 2^{m-1} },\quad
s \in \{-1,1\}^m
\label{eq:card-bd}
\end{align}
with conditional probability at least $1 - 2^{m+2} \exp( - \beta^2 n / 2^{m+8})$.
\end{itemize}
\end{lemma}

\begin{lemma}[Correlation of second-generation sets]
\label{lem:cor-2}
For any $\delta>0$ there is $n_0(\delta)$ with the following property. Let $n\geq n_0$, and assume \eqref{eq:conditions}, \eqref{eq:def-beta} and \eqref{eq:def-m}.
Then, conditional on any realization of $\{\SA,\SA',G(\SA),G(\SA'),\SB,\SB'\}$
such that \eqref{eq:cor-1}, \eqref{eq:b-size} 
hold, the second generation of partitions satisfies 
\begin{equation}
\label{eq:cor-bd}
\begin{split}
\Big| \Big\{ s \in \{-1, 1\}^m : |R_s \cap R'_s| \le \beta^2 n / 2^{m+6} \Big\} \Big| \le C \, 2^{m} m^{3/2} \, \beta^{1/4} \log(1/\beta)
\end{split}
\end{equation}
with conditional probability at least 
$1 - \exp ( - \beta^{9/4}n/2^m )$.
\end{lemma}

The proofs of the above two lemmas are postponed to Section~\ref{subs: sg corr}.

\subsection{Sparsification}
\label{sec:sparse}

Lemma~\ref{lem:cor-2} shows that with probability close to one,
for a large proportion of indices $s$, second-generation sets $R_s$ and $R_s'$ have large intersections.
This, in turn, implies (conditioned on a ``good'' realization of $\SA,$ $\SA',$ $\SB,$ $\SB',$
$G(\SA \cup \SB)$ and $G'(\SA' \cup \SB')$) that for a large proportion of $s\in\{-1,1\}^m$, and for any $i\in \SC\cap\SC'$,
the signs $\sign \big( |\neigh_{G}(i; R_s)|-p |R_s| \big)$ and $\sign \big( |\neigh_{G'}(i; R'_s)|-p |R'_s|\big)$
have non-negligible correlations. This suggests that the true permutation (the identity in our analysis)
can be recovered by comparing the ``full'' signatures $\big(\sign \big( |\neigh_{G}(i; R_s)|-p |R_s| \big)\big)_{s\in \{-1,1\}^m}$
and $\big(\sign \big( |\neigh_{G'}(j; R_s')|-p |R_s'| \big)\big)_{s\in \{-1,1\}^m}$.

However, this approach leads to a problem with resolving probabilistic dependencies, since for distinct $s, \tilde s \in \{-1, 1\}^m$,
$\sign \big( |\neigh_{G}(i; R_s)|-p |R_s| \big)$ and $\sign \big( |\neigh_{G'}(i; R_{\tilde s}')|-p |R_{\tilde s}'| \big)$
are {\it dependent} random variables as long as $R_s$ and $R_{\tilde s}'$ have a non-empty intersection.
Our bound from Lemma~\ref{lem:cor-2} still allows non-negligible intersections $R_s \cap R'_{\tilde s}$ for distinct indices $s,\tilde s \in \{-1, 1\}^m$, 
so that it is not clear what a comparison of the full signatures yields. 
As a way to deal with this issue, we perform {\it sparsification} of the collection of the second-generation sets.
Namely, we choose a uniform random subset $\cI \subset \{-1, 1\}^m$ of cardinality $\omega\ll 2^m$,
and construct vertex signatures on the basis of sets $R_s$ and $R_s'$ where $s\in \cI$.
It turns out that such a procedure efficiently handles the dependencies described above as the sparsified
collections $\{R_s : s\in \cI\}$ and $\{R_s' : s\in \cI\}$ have negligible intersections $R_s\cap R_{\tilde s}'$
for $s\neq \tilde s$. 
The following lemmas provide a rigorous justification for this step of the algorithm.


\begin{lemma}[Correlation of second-generation sets]
\label{lem:cor-3}
Condition on a realization of $\SA,$ $\SA',$ $\SB,$ $\SB',$ $\SC,$ $\SC'$, $G(\SA \cup \SB)$ and $G'(\SA' \cup \SB')$ such that \eqref{eq:cor-bd} holds with a constant $C>0$. 
Let $\cI$ be a uniform random subset of $\{-1,1\}^m$ of cardinality $\omega$. Then we have  
\begin{equation}
\label{eq:cor-bd-3}
\begin{split}
\Big| \Big\{ s \in \cI : |R_s \cap R'_s| \le \beta^2 n / 2^{m+6} \Big\} \Big| \le 2 C \omega m^{3/2} \, \beta^{1/4} \log(1/\beta)
\end{split}
\end{equation}
with probability at least 
$1 - 2 e^{ - (\log n)^{1+\delta} }$, where the randomness is with respect to $\cI$. 
\end{lemma}

In the sequel, for any subset $\cJ \subset \{-1,1\}^m$, we use the notation
\begin{align}
R_{\cJ} := \bigcup_{t \in \cJ} R_t , \quad
R'_{\cJ} := \bigcup_{t \in \cJ} R'_t . 
\label{eq:rj}
\end{align}

\begin{lemma}[Small overlaps]
\label{lem:spa-2}
Condition on a realization of $\SA, \SA', \SB, \SB', \SC, \SC'$, $G(\SA \cup \SB)$ and $G'(\SA' \cup \SB')$ such that 
$\frac{\beta n}{2^{m+3}} \le |R_s| ,  |R'_s| \le \frac{\beta n}{2^{m-1}}$ for all $s\in\{-1,1\}^m$. 
Let $\cI$ be a uniform random subset of $\{-1,1\}^m$ of cardinality $\omega$ as in the algorithm. 
Then, there exists a random subset $\tilde \cI \subset \cI$ measurable
with respect to $\{\SA, \SA', \SB, \SB', \SC, \SC',G(\SA \cup \SB),G'(\SA' \cup \SB'),\cI\}$,
such that the following holds with probability at least $1 - 2 n^2 e^{-(\log n)^{1+\delta}}$. 
We have 
\begin{align}
|\cI \setminus \tilde \cI| \le 4 (\log n)^{1+\delta} \label{eq:c-bd} ,
\end{align}
and for all $i \in \SC \cap \SC'$ and $s \in \tilde \cI$, 
\begin{align}
|\eta(i)_s| \lor |\eta'(i)_s| \le 
\sqrt{ \frac{ p n }{ 2^m (\log n)^{ \delta} } } 
\label{eq:eta-bd}
\end{align}
where
\begin{align}
\eta(i)_s &:= |\neigh_G(i; R_s \cap R'_{\cI \setminus \{s\}})| - p \cdot |R_s \cap R'_{\cI \setminus \{s\}}| ,
\label{eq:eta-s} 
\\
\eta'(i)_s &:=  |\neigh_{G'}(i; R'_s \cap R_{\cI \setminus \{s\}})| - p \cdot |R'_s \cap R_{\cI \setminus \{s\}}| .
\label{eq:eta-sp}
\end{align}
\end{lemma}

\begin{remark}
Note that, conditional on $\frac{\beta n}{2^{m+3}} \le |R_s| \le \frac{\beta n}{2^{m-1}}$, the standard deviation of the random variable
$|\neigh_G(i; R_s)| - p \cdot |R_s|$ is of order $\Theta(\sqrt{p|R_s|})=\Theta(\sqrt{\beta p n/2^m})$.
The above lemma asserts that on an event of probability close to one, and for a vast majority of indices $s\in \cI$, the deviation 
$|\neigh_G(i; R_s \cap R'_{\cI \setminus \{s\}})| - p \cdot |R_s \cap R'_{\cI \setminus \{s\}}|$
is of a much smaller order $\sqrt{ \frac{ p n }{ 2^m (\log n)^{ \delta} } }$. This should be interpreted as the property
that ``most of the randomness'' of $|\neigh_G(i; R_s)| - p \cdot |R_s|$ comes from the variable $|\neigh_G(i; R_s\setminus R_{\cI\setminus\{s\}}')|$.
This crucial property will allow us to compare the vertex signatures over $\{R_s : s\in \cI\}$ and $\{R_s' : s\in \cI\}$.
\end{remark}

The proofs of Lemmas~\ref{lem:cor-3} and~\ref{lem:spa-2} can be found in Section~\ref{1-9481-4098214-09}.

\subsection{Vertex signatures}
\label{sec:ver-sig}

In the sequel, we assume that $n\geq n_0$ for a sufficiently large $n_0=n_0(\delta) \in \N_+$ and
assume \eqref{eq:conditions}, \eqref{eq:def-beta}, \eqref{eq:def-m} and \eqref{eq:def-omega}. 
We continue to use the notation $R_{\cJ}$ and $R'_{\cJ}$ defined in \eqref{eq:rj}. 
Recall that $f(i)$ and $f'(i)$ denote the signatures of $i$ in $G$ and $G'$ respectively, defined by \eqref{eq:sig}. 
The next result controls the difference between them. 

\begin{lemma}[Correlated signatures]
\label{lem:cor-sig}
There exist universal constants $c_1, c_2 > 0$ such that the following holds. 
Condition on a realization of $\SA, \SA', \SB, \SB', \SC, \SC'$, $G(\SA \cup \SB)$ and $G'(\SA' \cup \SB')$ such that
\eqref{eq:b-size}, \eqref{eq:card-bd} and \eqref{eq:cor-bd} hold. 
Fix any vertex $i \in \SC \cap \SC'$.  
Condition further on a realization of $\cI \subset \{-1, 1\}^m$, $\tilde \cI \subset \cI$, 
and $\{ \neigh_G (i; R_s \cap R'_t ) , \, \neigh_{G'} (i; R_s \cap R'_t ) : s, t \in \cI, \, s \ne t \}$ such that \eqref{eq:cor-bd-3}, \eqref{eq:c-bd} and \eqref{eq:eta-bd} hold. 
Then we have 
\begin{align*}
\sum_{s \in \cI} \1 \big\{ f(i)_s = f'(i)_s \big\} \ge \frac{\omega}{2} (1 + c_1 \beta )
\end{align*}
with conditional probability at least $1 - \exp( - c_2 \beta^2 \omega)$. 
\end{lemma}

The following lemma, on the other hand, compares vertex signatures $f(i)$ and $f'(j)$ for distinct $i, j \in [n]$. 

\begin{lemma}[Uncorrelated signatures]
\label{lem:uncor-sig}
Condition on any realization of $\SA, \SA', \SB, \SB', \SC, \SC'$, $G(\SA \cup \SB)$ and $G'(\SA' \cup \SB')$
such that \eqref{eq:card-bd} holds. 
Fix any distinct vertices $i, j \in \SC \cap \SC'$.  
Condition further on a realization of $\cI \subset \{-1, 1\}^m$, $\tilde \cI \subset \cI$, 
and 
\begin{align}
\{ \neigh_G (i; R_s \cap R'_t ) , \, 
\neigh_{G'} (j; R_s \cap R'_t ) 
: s, t \in \cI, \, s \ne t \}
\label{eq:neigh-cond}
\end{align}
such that \eqref{eq:c-bd} holds, and \eqref{eq:eta-bd} holds for $i$ and also for $j$ in place of $i$.  
Then we have
\begin{align}
\sum_{s \in \cI} \1 \big\{ f(i)_s = f'(j)_s \big\} \le \frac{\omega}{2} \Big( 1 +  \frac{ \beta }{\log \log n} \Big)
\label{eq:uncor}
\end{align}
with conditional probability at least $1 - \exp \big( \frac{- \beta^2 \omega}{   8 (\log \log n)^{2} } \big)$.
\end{lemma}

The proofs of the above two lemmas can be found in Section~\ref{108746104987}. 

\subsection{Proof of Theorem~\ref{thm:main}}\label{1874601847608}
Given the above results established, the proof of the theorem consists of proving several simple claims.

\medskip
\noindent
{\bf Claim 1.} We have
\begin{align*}
&\Prob\Big\{ \sum_{s \in \cI} \1 \big\{ f(i)_s = f'(j)_s \big\} \le \frac{\omega}{2} \Big( 1 +  \frac{ \beta }{\log \log n} \Big) \mbox{ for all distinct $i , j \in \SC \cap \SC'$}\Big\} \\
&\geq 1- n^{-(\log n)^{\delta/2}}.
\end{align*}
Indeed, by combining Lemmas~\ref{lem:sym-diff-1} and~\ref{lem:card-2}, we get that
with probability at least
$$1 - 5 n^3 \exp(-\beta^3 \sqrt{pn} ) - 2 \exp( - \beta^2 n /40 ) - 2^{m+2} \exp( - \beta^2 n / 2^{m+8}) ,$$
\eqref{eq:b-size} and \eqref{eq:card-bd} hold. 
Furthermore, conditional on such a realization of $\{ R_s, R'_s : s \in \{-1,1\}^m\}$, 
Lemma~\ref{lem:spa-2} implies that \eqref{eq:c-bd} and \eqref{eq:eta-bd} hold with conditional probability at least 
$$1 - 2 n^2 e^{-(\log n)^{1+\delta}}.$$ 
Finally, conditional on such a realization of $\cI$ and $\neigh_G(i; R_s \cap R'_t)$, $\neigh_{G'}(i; R_s \cap R'_t)$ for $i \in \SC \cap \SC'$ and distinct $s,t \in \cI$, we apply Lemma~\ref{lem:uncor-sig} to obtain
that with conditional probability at least 
$$
1 -  n^2 \exp \Big( \frac{- \beta^2 \omega}{   8 (\log \log n)^{2} } \Big) ,
$$
\eqref{eq:uncor} holds for every pair of distinct
vertices $i,j\in \SC \cap \SC'$. It remains to note that, in view of our assumptions on the parameters,
the sum of the failure probabilities is bounded by $n^{-(\log n)^{\delta/2}}$. 

\medskip
\noindent
{\bf Claim 2.} We have 
\begin{align*}
\Prob\Big\{ \sum_{s \in \cI} \1 \big\{ f(i)_s = f'(i)_s \big\} > \frac{\omega}{2} \Big( 1 +  \frac{ \beta }{\log \log n} \Big) \mbox{ for all $i \in \SC \cap \SC'$}\Big\}
\geq 1- n^{-(\log n)^{\delta/2}}.
\end{align*}
Similarly, we first condition on any realization of $\{\SA,\SB,\SC,G(\SA\cup\SB),\SA',\SB',\SC',G'(\SA'\cup\SB')\}$ such that
\eqref{eq:b-size}, \eqref{eq:card-bd} and \eqref{eq:cor-bd} hold. 
Then we condition on a realization of $\cI$ and $\neigh_G(i; R_s \cap R'_t)$, $\neigh_{G'}(i; R_s \cap R'_t)$ for $i \in \SC \cap \SC'$ and distinct $s,t \in \cI$ such that \eqref{eq:cor-bd-3}, \eqref{eq:c-bd} and \eqref{eq:eta-bd} hold. 
Finally, we apply Lemma~\ref{lem:cor-sig} to obtain  
$$\sum_{s \in \cI} \1 \big\{ f(i)_s = f'(i)_s \big\} \ge \frac{\omega}{2} (1 + c_1 \beta ) > \frac{\omega}{2} \Big( 1 +  \frac{ \beta }{\log \log n} \Big) \mbox{ for all $i \in \SC \cap \SC'$}$$
with conditional probability at least $1 - n \exp( - c_2 \beta^2 \omega )$. 
It remains to note that 
the sum of the failure probabilities from Lemmas~\ref{lem:sym-diff-1}, \ref{lem:card-2}, \ref{lem:cor-2}, \ref{lem:cor-3}, \ref{lem:spa-2} and~\ref{lem:cor-sig} is bounded by 
\begin{align*}
&5 n^3 \exp(-\beta^3 \sqrt{pn} ) 
+ 2 \exp( - \beta^2 n /40 ) 
+ 2^{m+2} \exp( - \beta^2 n / 2^{m+8}) + 
\exp ( - \beta^{9/4}n/2^m ) \\
&+ 2 e^{ - (\log n)^{1+\delta} } 
+ 2 n^2 e^{-(\log n)^{1+\delta}}
+ n \exp( - c_2 \beta^2 \omega )
\le n^{ - (\log n)^{\delta/2} } ,
\end{align*}
in view of our assumptions on the parameters.

\medskip
\noindent
{\bf Claim 3.} Let $\mathcal T^h_u$, $h=1,2$, $u=1,2,\dots,\lfloor (\log n)^2/\beta^4\rfloor$, be independent uniform random
subsets of $[n]$ of cardinality $0.5\beta n$. Then
$$
\Prob\big\{\mbox{For every pair of indices $i,j$ there is $u$ such that $\{i,j\}\subset \mathcal T^1_u\cap \mathcal T^2_u$}\big\}
\geq 1-n^{-c\log n}$$
for a universal constant $c>0$.

Indeed, the probability that any given $i,j$ belong to $\mathcal T^1_1\cap \mathcal T^2_1$ can be bounded from below by
$\beta^4 /17$, whence the event of interest holds with probability at least $1-(1-\beta^4/17)^{\lfloor (\log n)^2/\beta^4\rfloor}$,
implying the estimate.

\medskip

Let us conclude the proof now. 
Recall the decision rule \eqref{eq:decide} for whether we potentially match a pair of vertices from the two graphs $G$ and $G'$ in Step~5 of the algorithm.  
By the first two claims, repeating Steps 1--5 of the algorithm $(\log n)^2/\beta^4$ times, we will get with probability at least
$1-2(\log n)^2/\beta^4\,n^{-\log^{\delta/2}n}$ that, every two distinct vertices $i, j$ are not potentially matched
whenever they both fall into $\SC \cap \SC'$ at some iteration, and every vertex $i$ is matched with itself whenever it falls into
$\SC \cap \SC'$ at some iteration.
Moreover, the third claim above implies that, with overwhelming probability, every pair of vertices falls into $\SC \cap \SC'$ at least in one of the iterations. Thus, with probability at least
$$
1-2(\log n)^2/\beta^4\,n^{-\log^{\delta/2}n}-n^{-c\log n}\geq 1-n^{-\log^{\delta/4}n} ,
$$
the algorithm returns the true permutation.

%% file: appendix.tex

\section{Additional proofs}
\label{sec:add-pf}

\subsection{Lemmas for Erd\H{o}s--R\'enyi graphs}\label{subs: er aux lem}



Observations grouped together in this section are based on simple decoupling arguments and
classical concentration inequalities. Undoubtedly, some of them can be found in the literature,
and some can be easily strengthened by applying slightly more sophisticated arguments and
sharper concentration inequalities. We prefer to provide all proofs for completeness.

\begin{lemma}[Comparison of an empirical distribution and a binomial]\label{l: emp vs binom}
Let $k\in\N_+$, $p\in(0,1)$; let $\Gamma$ be a $G(k,p)$ random graph, $J$ be a fixed non-empty subset of $[k]$,
and $M$ be a possibly random subset of $[k]\setminus J$ measurable with respect to $\Gamma([k]\setminus J)$, the subgraph of $\Gamma$ induced by $[k] \setminus J$. 
Denote by $\tilde F$ the empirical distribution of the number of neighbors of vertices in $J$ within the set $M$, i.e.\
$$
\tilde F(t):=\frac{1}{|J|}\big|\big\{i\in J:\;|\neigh_\Gamma(i;M)|\leq t\big\}\big|,\quad t\in\N_0.
$$
Then for every $s>0$ we have with probability at least $1-2k\,\exp(-c s^2/|J|)$:
$$
\big|\tilde F(t)-F_{|M|,p}(t)\big|\leq \frac{s}{|J|},\quad t\in\N_0,
$$
where $c>0$ is a universal constant.
\end{lemma}
\begin{proof}
Let us condition on any realization of $M$ (in what follows, by $\tilde \Prob$ we denote the corresponding conditional probability measure).
Fix for a moment any $t\in\N_0$, and let, for each $i\in J$, $b_i$ be the indicator of the
event that $|\neigh_\Gamma(i;M)|\leq t$. Then, clearly, $b_i$'s are mutually independent, each with probability of
success $F_{|M|,p}(t)$. Hence, applying Hoeffding's inequality,
$$
\tilde\Prob\bigg\{\Big|\sum_{i\in J}b_i-|J|\,F_{|M|,p}(t)\Big|\geq s\bigg\}\leq 2\exp(-c s^2/|J|),\quad s\geq 0,
$$
for a universal constant $c>0$.
Taking the union bound over all $t\in\{0,1,\dots,k-1\}$, we obtain the result.
\end{proof}

\begin{lemma}[Anti-concentration of vertex degrees]\label{lem:non-conc}
Let $d, r \in \N_0$, $p\in(0,1)$,
and $\Gamma$ be a $G(k,p)$ random graph. Then with probability at least
$1-2k^2\,\exp(-\frac{cr^3}{p(1-p)k})$, we have
$$
\big|\big\{i \in [k] :\;\deg_\Gamma(i)\in(d,d+r]\big\}\big|\leq Cr\,\sqrt{\frac{k}{p(1-p)}}.
$$
Here, $c,C>0$ are universal constants.
\end{lemma}
\begin{proof}
If $k\leq r$ then $r\,\sqrt{\frac{k}{p(1-p)}}\geq k$, and the above assertion is trivial. Further, we assume that $k>r$.
Set $h:=\lfloor k/r\rfloor$ and $s:=r^2/\sqrt{p(1-p)k}$.
Denote
$$
T:=\big\{i \in [k] :\;\deg_\Gamma(i)\in(d,d+r]\big\}.
$$
Let $J_1,J_2,\dots,J_h$ be any fixed partition of $[k]$ into $h$ sets, each of cardinality locked in the interval $[k/(2h),2k/h]$.
Obviously, the set $T$ is contained in the union
$$
\bigcup\limits_{j=1}^h \big\{i\in J_j:\;|\neigh_\Gamma(i;[k]\setminus J_j)|\in (d-|J_j|;d+r]\big\}.
$$
Fix for a moment any $j\leq h$. We have
$$
\big|\big\{i\in J_j:\;|\neigh_\Gamma(i;[k]\setminus J_j)|\in (d-|J_j|;d+r]\big\}\big|
=|J_j|(\tilde F(d+r)-\tilde F(d-|J_j|)),
$$
where $\tilde F$ is the empirical distribution of the number of neighbors of vertices in $J_j$ within $[k]\setminus J_j$, given by
$$
\tilde F(t)=\frac{1}{|J_j|}\big|\big\{i\in J_j:\;|\neigh_\Gamma(i;[k]\setminus J_j)|\leq t\big\}\big|,\quad t\in\N_0.
$$
Applying Lemma~\ref{l: emp vs binom}, we obtain
$$
\big|(\tilde F(d+r)-\tilde F(d-|J_j|))-(F_{k-|J_j|,p}(d+r)-F_{k-|J_j|,p}(d-|J_j|))\big|\leq \frac{2s}{|J_j|}
$$
with probability at least 
$1-4k\,\exp(-c s^2/|J_j|)$.
In view of Lemmas~\ref{l: comp of binom} and~\ref{lem:bin-cdf-gap}, we then get for some universal constant $C>0$,
$$
\tilde F(d+r)-\tilde F(d-|J_j|)\leq F_{k,p}(d+r+|J_j|)-F_{k,p}(d-|J_j|)+\frac{2s}{|J_j|}
\leq \frac{C(r+k/h)}{\sqrt{p(1-p)k}}+\frac{2s}{|J_j|} .
$$

Summarizing,
$$
|T|\leq h\bigg(\frac{C(r+k/h)\cdot 2k/h}{\sqrt{p(1-p)k}}+2s\bigg)=2C\sqrt{\frac{k}{p(1-p)}}\,(r+k/h)+2hs
$$
with probability at least $1-4k^2\,\exp(-c hs^2/(2k))$. Our choice of $h$ and $s$ implies the result.
\end{proof}



Finally, we establish a lemma for the first-generation partition. 

\begin{lemma}[Cardinalities of first-generation sets]
\label{lem:card-1}
For $k \in \N_+$ and $p \in (0,1)$, let $\Gamma$ be a $G(k,p)$ random graph. 
For $m \in \N_+$, we define {\it the first generation with respect to $\Gamma$} as the partition of $[k]$ consisting of sets 
$$
Q^\Gamma_\ell:=\big\{j \in [k] :\;F_{k,p}( \deg_\Gamma(j) )\in\big((\ell-1)/m,\ell/m\big]\big\},\quad \ell \in [m] .
$$ 
Then, for $\beta \in(0,1)$, we have
$$
\Big||Q^\Gamma_\ell|-\frac{k}{m}\Big|\leq C k \sqrt{ \beta } ,
\quad \ell \in [m] 
$$
with probability at least $1-2mk^2\exp\big(-c\, \beta^{3/2} \sqrt{k p (1-p) } \, \big)$,
where $C,c>0$ are universal constants.
\end{lemma}

\begin{proof}
Define integers $d_{0},d_1,\dots,d_{m}$ via the condition
$$
F_{k,p}(t)\in\big((\ell-1)/m,\ell/m\big]\mbox{ if and only if }d_{\ell-1}< t\leq d_\ell\quad\mbox{for every $t\in\N_0$ and $\ell \in [m]$}. 
$$
Fix any $\ell \in [m]$.
Set $h:= \Big\lfloor \sqrt{ \frac{k}{\beta p(1-p)} } \, \Big\rfloor$ and $s:= \beta \sqrt{ k p(1-p) } $.
Note that if $h\geq k$ then necessarily $\sqrt{k\beta p (1-p) }\leq 1$, and the assertion of the statement is trivial.
Further, we assume that $h<k$.
Let $J_1,J_2,\dots,J_h$ be any fixed partition of $[k]$ into $h$ sets, each of cardinality locked in the interval $[k/(2h),2k/h]$.
Clearly, we have
$$
J_u\cap Q^\Gamma_\ell\subset\big\{j\in J_u:\;|\neigh_{\Gamma}(j;[k]\setminus J_u)|\in (d_{\ell-1}-|J_u|,d_\ell]\big\},\quad u \in [h] .
$$
Applying Lemma~\ref{l: emp vs binom}, as well as Lemmas~\ref{l: comp of binom} and~\ref{lem:bin-cdf-gap}, we get
\begin{align*}
\frac{1}{|J_u|}\big|J_u\cap Q^\Gamma_\ell\big|
&\leq F_{k-|J_u|,p}(d_\ell)-F_{k-|J_u|,p}(d_{\ell-1}-|J_u|)+\frac{2s}{|J_u|}\\
&\leq F_{k,p}(d_\ell+|J_u|)-F_{k,p}(d_{\ell-1}-|J_u|)+\frac{2s}{|J_u|}\\
&\leq \frac{1}{m}+\frac{C|J_u|}{\sqrt{p(1-p)k}}+\frac{2s}{|J_u|},\quad u \in [h],
\end{align*}
with probability at least $1-4k^2\,\exp(-c' h s^2/k)$, for a universal constant $c'>0$. Thus,
$$
\big|Q^\Gamma_\ell\big|
\leq \frac{k}{m}+\frac{C'k^2/h}{\sqrt{p(1-p)k}}+2hs
$$
with probability at least $1-4k^2\,\exp(-c' h s^2/k)$. The lower bound is analogous, so the conclusion follows by plugging in $h$ and $s$ and applying a union bound.
\end{proof}

%
%
%
%

\subsection{Correlation of first-generation partitions}\label{subs: fg corr}


To study the correlation of the first-generation partitions in the algorithm, we start with a concentration result in a general setup. 

If $\Gamma$ is a $G(k,p)$ random graph, $J$ is a subset of vertices of $\Gamma$, and $i\in J^c$
is any vertex, then the difference $| \neigh_\Gamma (i; J)|  - p |J|$ should have a typical order of magnitude
$\sqrt{p|J|}$. The next lemma asserts that if $\Gamma$ and $\Gamma'$ are two $G(k,p)$ graphs generated according to the correlated Erd\H{o}s--R\'enyi graph model, 
then for any given set $J$ and vertex $i\in J^c$, the differences $| \neigh_\Gamma (i; J)|  - p |J|$
and $| \neigh_{\Gamma'} (i; J)|  - p |J|$ fluctuate ``synchronously'' in the sense that 
$$
\big(| \neigh_\Gamma (i; J)|  - p |J|\big)-\big(| \neigh_{\Gamma'} (i; J)|  - p |J|\big)
$$
is much smaller than $\sqrt{p|J|}$ with high probability. In fact, to provide necessary flexibility, we will allow different
vertex subsets $J$ and $J'$ for $\Gamma$ and $\Gamma'$, respectively, and estimate the quantity of interest in terms of the
intersection $J\cap J'$ and the symmetric difference $J \triangle J'$. 

\begin{lemma}[Correlation of sizes of neighborhoods] 
\label{lem:neigh-diff}
Fix $k \in\N_+$, $p \in (0,1)$ and $\alpha \in [0, 1-p]$. From a $G(k, \frac{p}{1-\alpha})$ random graph $\Gamma_0$, we sample two independent copies of subgraphs $\Gamma$ and $\Gamma'$
by removing every edge of $\Gamma_0$ independently with probability $\alpha$. 
Fix subsets $J, J' \subset [k]$ and a vertex $i \in [k] \setminus (J \cup J')$. For any $t>0$, we have that with probability at least $1 - 6 \exp (- t) - 2\exp \big( \frac{ - p |J \cap J'|}{3 (1- \alpha)} \big)$, 
\begin{align*}
\Big| \big| \neigh_\Gamma (i; J) \big|  - p |J| - \big| \neigh_{\Gamma'} (i; J') \big| + p |J'| \Big| \le 4 \Big( t + \sqrt{ t \alpha p |J \cap J'| } + \sqrt{ t p |J \triangle J'| } \Big) .
\end{align*}
\end{lemma}

\begin{proof}
Let us partition the neighborhood of interest as
$$
\neigh_\Gamma(i ; J) = \neigh_\Gamma(i ; J \cap J') \cup \neigh_\Gamma(i ; J \setminus J') 
$$
and similarly for $\Gamma'$. As a result, we can apply the triangle inequality to obtain
\begin{align*}
&\Big| \Big( \big| \neigh_\Gamma (i; J) \big|  - p \big| J \big| \Big) - \Big( \big| \neigh_{\Gamma'} (i; J') \big| - p \big| J' \big| \Big) \Big| 
\le \Big| \big| \neigh_\Gamma (i; J \cap J') \big| - \big| \neigh_{\Gamma'} (i; J \cap J') \big| \Big| \\
& \hspace{5cm} + \Big| \big| \neigh_\Gamma(i ; J \setminus J') \big| - p \big| J \setminus J' \big| \Big|
+ \Big| \big| \neigh_{\Gamma'}(i ; J' \setminus J) \big| - p \big| J' \setminus J \big| \Big| 
\end{align*}
It suffices to bound each term on the right-hand side.

First, we can write $| \neigh_\Gamma (i; J \cap J') | = \sum_{j \in J \cap J'} X_j Y_j$ and $| \neigh_{\Gamma'} (i; J \cap J') | = \sum_{j \in J \cap J'} X_j Y_j'$, where $X_i \sim \Ber(\frac{p}{1-\alpha})$, $Y_i \sim \Ber(1-\alpha)$ and $Y_i' \sim \Ber(1-\alpha)$ are all independent Bernoulli random variables. As a result, we have 
$$
\big| \neigh_\Gamma (i; J \cap J') \big| - \big| \neigh_{\Gamma'} (i; J \cap J') \big| = \sum_{j \in J \cap J'} X_j ( Y_j - Y_j') . 
$$
Conditioned on any realization of $X_j$, the random variables $X_j(Y_j - Y_j')$ are independent, zero-mean and take values in $\{-1,0,1\}$. Hence Bernstein's inequality (see Lemma~\ref{lem: Bern ineq}) implies that 
\begin{align}
\p \Big\{ \Big| \sum_{j \in J \cap J'} X_j ( Y_j - Y_j') \Big| \ge 4t/3 + 2 \sigma \sqrt{t} \Big\} \le 2 \exp ( -t ) , 
\label{eq:Y-tail}
\end{align}
where 
$$
\sigma^2 = \sum_{j \in J \cap J'} \E \big[ X_j^2 (Y_j - Y_j')^2  \,|\, X_j \big] =  2 \alpha (1- \alpha) \sum_{j \in J \cap J'} X_j . 
$$
In addition, since $\sum_{j \in J \cap J'} X_j \sim \Bin(|J \cap J'|, \frac{p}{1-\alpha})$, Bernstein's inequality gives
\begin{align}
\p \Big\{ \sum_{j \in J \cap J'} X_j \ge \frac{2p |J \cap J'|}{1-\alpha} \Big\} \le 2\exp \Big( \frac{ - p  |J \cap J'|}{3 (1- \alpha)} \Big) . 
\label{eq:X-tail}
\end{align}
Combining \eqref{eq:Y-tail} and \eqref{eq:X-tail} yields 
$$
\p \Big\{ \Big| \sum_{j \in J \cap J'} X_j ( Y_j - Y_j') \Big| \ge 4t/3 + 4 \sqrt{ t \alpha p |J \cap J'| } \Big\} \le 2 \exp (- t) + 2\exp \Big( \frac{ - p |J \cap J'|}{3 (1- \alpha)} \Big) .
$$

Second, since $|\neigh_\Gamma (i; J \setminus J')| \sim \Bin(|J \setminus J'|, p)$, it follows from Bernstein's inequality that 
$$
\p \Big\{ \Big| \big|\neigh_\Gamma (i; J \setminus J')\big| - p \big|J \setminus J'\big| \Big| \ge 4t/3 + 2 \sqrt{ t p |J \setminus J'| } \Big\} \le 2 \exp(-t) .
$$
An analogous bound holds for the graph $\Gamma'$. 
Combining everything finishes the proof.
\end{proof}

Recall that in the algorithm, $\SA$ and $\SA'$ are defined to be independent uniformly random subsets of $[n]$ of cardinality $n - \beta n$, where $\beta \ge \alpha$. 
We have the following result. 

\begin{lemma}[Overlap of vertices with similar degrees] 
\label{lem:inclusion-2}
For any $t>0$, we write 
\begin{align}
\tau := 4t + 12 \sqrt{ t \beta p n } ,
\quad
\upsilon := 6 \exp(-t) + 2\exp (-pn/4) 
\label{eq:td}
\end{align}
for brevity. 
Then with probability at least $1 - n \exp ( - \beta^2 \tau ) $, there exists a set $T \subset [n]$ of cardinality $|T| \le \upsilon n + 2 \beta n$ such that the following holds: 
For any $d, \bar d \in \N_0$ such that $d < \bar d$, 
\begin{align*}
\big\{ i \in \SA : \deg_{G(\SA)} (i) \in [ d +  3 \tau , \bar d - 3 \tau ) \big\} 
\subset \big\{i \in \SA' : \deg_{G'(\SA')}(i)  \in [ d , \bar d \, ) \big\} \cup T . 
\end{align*}
\end{lemma}

\begin{proof}
Note that without loss of generality, $4\leq \tau\leq n/3$: Indeed, if $\tau<4$ then necessarily $t\leq 1$, whence $\upsilon n>n$;
further, if $\tau>n/3$ then the set $\{ i \in \SA : \deg_{G(\SA)} (i) \in [ d +  3 \tau , \bar d - 3 \tau )\}$ is empty for any $d,\bar d\in \N_0$.
Let $h := \lfloor n / \tau \rfloor$. 

Note also that $|\SA \cap \SA'| \ge n - 2 \beta n$. 
Hence, 
we can partition $\SA \cap \SA'$ into $h$ (random) sets $I_1, \dots, I_h$, each of which has cardinality in the interval $[\tau/2, 2\tau]$. 

Now we fix $j \in [h]$. For each $i \in I_j$, define an indicator
$$
b_i := \1 \Big\{ \Big| \big| \neigh_{G} (i; \SA \setminus I_j) \big| - \big| \neigh_{G'} (i; \SA' \setminus I_j) \big| \Big| \ge \tau \Big\} .
$$   
To bound $\E[b_i] = \p\Big\{ \Big| \big| \neigh_{G} (i; \SA \setminus I_j) \big| - \big| \neigh_{G'} (i; \SA' \setminus I_j) \big| \Big| \ge \tau \Big\}$, we apply Lemma~\ref{lem:neigh-diff} with $\Gamma = G$, $\Gamma' = G'$, $J = \SA \setminus I_j$ and $J' = \SA' \setminus I_j$.  
In this case, we have $|J| = |J'|$, $0.9n \le |J \cap J'| \le n$ and $|J \triangle J'| \le 2 \beta n$, so Lemma~\ref{lem:neigh-diff} gives
$$
\E[b_i] \le \upsilon .
$$
The indicators $b_i$ are clearly independent, so Hoeffding's inequality yields that 
$$
\p \Big\{ \sum_{i \in I_j} b_i \ge \upsilon |I_j| + \beta \tau \Big\} \le \exp ( - 2 \beta^2 \tau^2 / |I_j| ) 
\le \exp( - \beta^2 \tau ) . 
$$

Next, we define a set $T_j := \{i \in I_j : b_i = 1\}$, which has cardinality $|T_j| \le \upsilon |I_j| + \beta \tau$ with probability at least $1 - \exp(- \beta^2 \tau).$ 
By the definitions of $b_i$ and $T_j$, we have
\begin{equation}\label{eq:inc}
\begin{split}
&\big\{i \in I_j : \big| \neigh_{G} (i; \SA \setminus I_j) \big| \in \big[ d + \tau, \bar d - 3 \tau \big) \big\}\\
&\hspace{2cm} \subset \big\{i \in I_j : \big| \neigh_{G'}(i; \SA' \setminus I_j) \big|  \in \big[ d , \bar d -  2 \tau \big) \big\} \cup T_j  . 
\end{split}
\end{equation}
Furthermore, in view of the bounds
$$
0 \le \deg_{G(\SA)}(i) - \big| \neigh_{G} (i; \SA \setminus I_j) \big| \le |I_j| 
\le 2 \tau
$$
and the analogous bounds for the graph $G'$, the inclusion \eqref{eq:inc} implies
\begin{align*}
\big\{i \in I_j : \deg_{G(\SA)} (i) \in \big[ d + 3 \tau , \bar d - 3 \tau \big) \big\} 
\subset \big\{i \in I_j : \deg_{G'(\SA')}(i)  \in [ d , \bar d \, ) \big\} \cup T_j . 
\end{align*}

Taking a union bound over $j \in [h]$, we see that
$\cup_{j=1}^h T_j$ has cardinality at most $\upsilon n + \beta n$ with probability at least $1 - n \exp ( - \beta^2 \tau ) $. 
Moreover, we have
\begin{align*}
&\big\{i \in \SA \cap \SA' : \deg_{G(\SA)} (i) \in [ d + 3 \tau , \bar d - 3 \tau ) \big\}  \\
&\subset \big\{i \in \SA \cap \SA' : \deg_{G'(\SA')}(i)  \in [ d , \bar d \, ) \big\} \cup \big( \cup_{j=1}^h T_j \big) . 
\end{align*}
Finally, it suffices to define $T := \big( \cup_{j=1}^h T_j \big) \cup (\SA \setminus \SA')$ and use the bound $|\SA \setminus \SA'| \le \beta n$ to complete the proof.
\end{proof}

We are ready to prove Lemma~\ref{lem:sym-diff-1}, which is our main result about the correlation between the first-generation partitions. 

\medskip

\begin{proof}[Proof of Lemma~\ref{lem:sym-diff-1}]
For each $\ell \in [m]$, let $d_\ell$ be the smallest integer $d$ such that $F_{|\SA|,p} ( d ) > \frac{\ell - 1}{m}$, and let $d_{m + 1} = \infty$. 
Then the set $Q_\ell$ in \eqref{eq:first-gen} can be equivalently defined by
$$
Q_\ell =  \big\{i \in \SA : \deg_{G(\SA)}(i) \in [ d_\ell, d_{\ell + 1}  ) \big\} ,
$$
and similarly for $Q'_\ell$. 

Let $t = \log(1/\beta)$, and
let $\tau$ and $\upsilon$ be defined by \eqref{eq:td}. Observe that, with our assumptions on parameters, $2\exp(-pn/4)\ll \beta$, so
\begin{equation}\label{eq: aux 0245209}
\upsilon\leq 7\beta.
\end{equation}
Further, assuming $n_0(\delta)$ is large enough, we have
\begin{equation}\label{eq: aux 31241}
\beta\sqrt{pn}\ll \tau=4\log(1/\beta) + 12 \sqrt{pn\beta \log(1/\beta)}\leq 12.01 \sqrt{pn\beta \log(1/\beta)}.
\end{equation}
Lemma~\ref{lem:non-conc} with $\Gamma = G(\SA)$, $r = \lfloor 3 \tau\rfloor$ and a union bound yields the following: With probability at least $1 - 2 n^3 \exp \big( \frac{ - c \tau^3 }{ p n } \big)$, it holds for all 
$d = 0, 1, \dots, |\SA|$ that
$$
\big|\big\{i \in \SA : \deg_{G(\SA)} (i)\in [d, d +  3 \tau) \big\}\big|\leq C \tau \sqrt{n/p} .
$$
This bound applied with $d = d_\ell$ and $d = d_{\ell+1} - 3 \tau$, gives
$$
\big|\big\{i \in \SA : \deg_{G(\SA)} (i)\in [d_\ell, d_\ell +  3 \tau)\cup[d_{\ell+1} - 3 \tau,d_{\ell+1}) \big\}\big|\leq 2C \tau \sqrt{n/p},\quad
\ell\in[m],
$$
with probability at least $1 - 2 n^3 \exp \big( \frac{ - c \tau^3 }{ p n } \big)$.
Together with Lemma~\ref{lem:inclusion-2} applied with $d = d_\ell$ and $\bar d = d_{\ell+1}$, this implies that with probability at least $1 - 2 n^3 \exp \big( \frac{ - c \tau^3 }{ p n } \big) - n \exp( - \beta^2 \tau )$, 
\begin{align}
\big\{i \in \SA : \deg_{G(\SA)}(i) \in [ d_\ell, d_{\ell + 1}  ) \big\} \subset \big\{i \in \SA' : \deg_{G'(\SA')}(i)  \in [ d_\ell , d_{\ell + 1} ) \big\} \cup T\cup T_\ell 
\label{eq:inc2}
\end{align}
for all $\ell \in [m]$, where $T$ is a random set of cardinality 
$
|T| \le \upsilon n + 2 \beta n,
$ 
and for each $\ell\in[m]$, $T_\ell$ is a random set of cardinality $|T_\ell|\leq 2 C \tau \sqrt{n/p}$.

The inclusion \eqref{eq:inc2} is equivalent to $Q_\ell \setminus Q'_\ell \subset T\cup T_\ell$.
Therefore, by symmetry, we can control $| Q_\ell \triangle Q'_\ell |$. Specifically, 
plugging in $\tau$ and $\upsilon$ from \eqref{eq:td} and in view of \eqref{eq: aux 0245209},~\eqref{eq: aux 31241}, we obtain 
$$
\big| Q_\ell \triangle Q'_\ell \big| \le  
C' n \sqrt{ \beta \log(1/\beta) }, \quad
\ell \in [m] 
$$
with probability at least $1 - 4 n^3 \exp(-\beta^3 \sqrt{pn} )$. 
Furthermore, Lemma~\ref{lem:card-1} with $\Gamma = G(\SA)$ gives
$$
\Big| |Q_\ell| - \frac{n - \beta n}{ m } \Big| \le C'' n \sqrt{\beta} ,
\quad \ell \in [m] 
$$
with probability at least $1- n^3 \exp(- \beta^2 \sqrt{n p}  )$. 
The conclusion readily follows.
\end{proof}

\subsection{Correlation of second-generation partitions}\label{subs: sg corr}

We now study the correlation between the second-generation partitions. 

\medskip

\begin{proof}[Proof of Lemma~\ref{lem:card-2}]
First, recall that $|\SB| = |\SB'| = 0.5 \beta n$. 
By the definitions of $\SB$ and $\SB'$, we can view $|\SB \cap \SB'|$ as a sum of $0.5 \beta n$ indicators sampled without replacement from a population of $0.5 \beta n$ ones and $n - 0.5 \beta n$ zeros. Therefore, Bernstein's inequality for sampling without replacement (Lemma~\ref{lem:bern-wo}) yields that
\begin{align*}
\p \big\{ \big| |\SB \cap \SB'| - \beta^2 n / 4 \big| \ge \beta^2 n / 8 \big\} \le  2 \exp( - \beta^2 n /40 ) .
\end{align*}

Next, we condition on a realization of $\SA, \SA', \SB, \SB'$ as well as $G(\SA)$ and $G'(\SA')$, so that the first assertion of the lemma, as well as \eqref{eq:cor-1} hold.
For $\ell \in [m]$ and $i \in \SB$, we have $|\neigh_{G}(i; Q_\ell)| \sim \Bin( |Q_\ell|, p )$. Since the median of a $\Bin(k,p)$ random variable is either $\lfloor kp \rfloor$ or $\lceil kp \rceil $, we have
$$
\Big| \tilde\p \big\{ \sign \big( |\neigh_{G}(i; Q_\ell)| - p |Q_\ell| \big) = 1 \big\} - 1/2 \Big| \le \frac{C}{ \sqrt{ p |Q_\ell| } } \le C' \sqrt{ \frac{m}{np} } 
$$
by Lemma~\ref{lem:bin-cdf-gap} and  \eqref{eq:cor-1} together with the condition that $1/m \gg \sqrt{\beta \log(1/\beta)}$.  
As a result, for any fixed $s \in \{-1,1\}^m$,
\begin{equation}\label{eq: aux 20920498}
\Big| \tilde\p \big\{ s_\ell = \sign \big( |\neigh_{G}(i; Q_\ell)| - p |Q_\ell| \big) \text{ for all } \ell \in [m] \big\} - 1/2^m \Big| \le \frac{ C m^{3/2} }{ \sqrt{np} } \le 1/2^{m+1} 
\end{equation}
by the assumptions on $n, p$ and $m$.  

In view of the definition of $R_s$ in \eqref{eq:sec-gen} and \eqref{eq: aux 20920498}, we can write
$
|R_s| = \sum_{i \in \SB} X_i ,
$
where $X_i$'s are independent Bernoulli random variables: $X_i \sim \Ber(b_i)$ with $b_i \in [ 1/2^{m+1}, 1/2^{m-1} ]$. 
Consequently, Bernstein's inequality gives 
$$
\tilde\p \Big\{ |\SB|/2^{m+2} \le |R_s| \le |\SB|/2^{m-2}  \Big\} \ge 1 - 2 \exp \Big( \frac{ -|\SB| }{ 2^{m+5} } \Big) .
$$
Similarly, we can also bound $|R_s \cap \SB'|$ as
$$
\tilde\p \Big\{ |\SB \cap \SB'| / 2^{m+2} \le |R_s \cap \SB'| \le |\SB \cap \SB'| / 2^{m-2}   \Big\} \ge 1 - 2 \exp \Big( \frac{ -|\SB \cap \SB'| }{ 2^{m+5} } \Big) .
$$
The conclusion follows by combining the above bounds with a union bound over $s \in \{-1,1\}^m$.
\end{proof}

Next, similarly to the first-generation sets, we prove that corresponding second-generation sets of $G$ and $G'$ have large intersections
with high probability. Our proof techniques allow us to show that this is true {\it on average},
for a large proportion of indices $s$ (see Lemma~\ref{lem:cor-2}).

\begin{lemma}[Probability that signs coincide]
\label{lem:sign-same}
For any $\delta>0$ there is $n_0(\delta)$ with the following property. Let $n\geq n_0$, and assume \eqref{eq:conditions}, \eqref{eq:def-beta} and \eqref{eq:def-m}.
Then, conditional on any realization of $\{\SA,\SA',G(\SA),G'(\SA')\}$ such that \eqref{eq:cor-1} holds, 
and conditional further on any realization of $\SB, \SB'$ with $\SB \cap \SB'\neq\emptyset$, we have
$$
\tilde\p \Big\{ \sign \big( |\neigh_{G}(i; Q_\ell)| - p |Q_\ell| \big) = \sign \big( |\neigh_{G'}(i; Q'_\ell)| - p |Q'_\ell| \big) \Big\} \ge 1 - C \sqrt{m} \, \beta^{1/4} \log(1/\beta) 
$$
for every vertex $i \in \SB \cap \SB'$, for all $\ell \in [m]$, and
for a universal constant $C>0$. Here, $\tilde\p$ denotes the conditional probability measure. 
\end{lemma}

\begin{proof}
%
Fix any vertex $i \in \SB \cap \SB'$. We shall apply Lemma~\ref{lem:neigh-diff} with $\Gamma = G$, $\Gamma' = G'$, $J = Q_\ell$, $J' = Q'_\ell$ and $t = \log(1/\beta)$.
Note that, with our choice of parameters, we have (deterministically) that
$$
t + \sqrt{ t \alpha p |J \cap J'| } + \sqrt{ t p |J \triangle J'| }
\leq \log(1/\beta)+\sqrt{pn\beta\log(1/\beta)}+\sqrt{ t p |J \triangle J'| },
$$
and that, in view of our conditioning such that \eqref{eq:cor-1} holds, we have
$$
\sqrt{ t p |J \triangle J'| }\leq C'\sqrt{pn\log(1/\beta)}(\beta \log(1/\beta))^{1/4}.
$$
Hence, applying Lemma~\ref{lem:neigh-diff}, \eqref{eq:cor-1}, and the condition that $2 \exp(\frac{-pn}{8m}) \ll \beta$, we see that with conditional probability at least $1 - 7 \beta$, 
\begin{align*}
\Big| \Big( \big| \neigh_G (i; Q_\ell) \big| - p|Q_\ell| \Big) - \Big( \big| \neigh_{G'} (i; Q'_\ell) \big| - p|Q'_\ell| \Big) \Big| \le C \sqrt{ p n } \, \beta^{1/4} \log(1/\beta) 
\end{align*}
for a universal constant $C>0$. 
Moreover, $|\neigh_G (i; Q_\ell)| \sim \Bin( |Q_\ell|, p )$, so Lemma~\ref{lem:bin-cdf-gap} yields 
$$
\tilde\p \big\{ \big| |\neigh_G (i; Q_\ell)| - p |Q_\ell| \big| > C \sqrt{ p n } \, \beta^{1/4} \log(1/\beta) \big\} > 1 - C' \sqrt{m} \, \beta^{1/4} \log(1/\beta) 
$$
for a universal constant $C'>0$. 
Combining the above bounds finishes the proof.
\end{proof}

We are ready to prove Lemma~\ref{lem:cor-2}. 

\medskip

\begin{proof}[Proof of Lemma~\ref{lem:cor-2}]
Take any vertex $i \in \SB \cap \SB'$. If $\sign \big( |\neigh_{G}(i; Q_\ell)| - p |Q_\ell| \big) = \sign \big( |\neigh_{G'}(i; Q'_\ell)| - p |Q'_\ell| \big)$, $\ell\in[m]$, then $i$ belongs to $R_s \cap R'_s$ for some $s \in \{-1,1\}^m$ by definition. Therefore,
Lemma~\ref{lem:sign-same} together with Bernstein's inequality implies that
\begin{align}
\sum_{s \in \{-1,1\}^m} |R_s \cap R'_s| \ge |\SB \cap \SB'| \Big( 1 - C m^{3/2} \, \beta^{1/4} \log(1/\beta)  \Big) 
\label{eq:ineq-1}
\end{align}
with conditional probability at least $1 - \exp \big( -c |\SB \cap \SB'| m^{3/2} \, \beta^{1/4} \log(1/\beta) \big) $, where $C, c>0$ are universal constants.

Next, conditional on a realization of $G(\SA\cup \SB)$ and $G(\SA'\cup\SB')$ such that \eqref{eq:ineq-1} and \eqref{eq:card-bd} hold, 
it follows from the bound $|R_s \cap \SB'| \ge \beta^2 n/2^{m+5}$ that
\begin{align*}
&\sum_{s \in \{-1,1\}^m} (\beta^2 n / 2^{m+6} ) \cdot \1 \big\{ |R_s \cap R'_s| \le \beta^2 n / 2^{m+6} \big\} \\
&\le \sum_{s \in \{-1,1\}^m} \Big( |R_s \cap \SB' | - |R_s \cap R'_s | \Big) \cdot \1 \big\{ |R_s \cap R'_s| \le \beta^2 n / 2^{m+6} \big\} \\
&\le \sum_{s \in \{-1,1\}^m} \Big( |R_s \cap \SB' | - |R_s \cap R'_s | \Big)  
=  |\SB \cap \SB'| - \sum_{s \in \{-1,1\}^m} |R_s \cap R'_s| . 
\end{align*}
Together with \eqref{eq:ineq-1}, this yields 
\begin{align*}
\sum_{s \in \{-1,1\}^m} \1 \big\{ |R_s \cap R'_s| \le \beta^2 n / 2^{m+6} \big\} 
&\le (2^{m+6}/ \beta^2 n) \cdot |\SB \cap \SB'| \cdot  C m^{3/2} \, \beta^{1/4} \log(1/\beta) . 
\end{align*}
Finally, we note that $|\SB \cap \SB'| \le \beta^2 n/ 2$ on our event.

Hence, conditional on a realization of $\{\SA,\SA',G(\SA),G(\SA'),\SB,\SB'\}$
such that \eqref{eq:cor-1} and \eqref{eq:b-size} 
hold,
the required condition on $|R_s \cap R'_s|$, $s\in\{-1,1\}^m$, is 
satisfied with conditional probability at least $1-\exp \big( -c |\SB \cap \SB'| m^{3/2} \, \beta^{1/4} \log(1/\beta) \big)
- 2^{m+2} \exp( - \beta^2 n / 2^{m+8})\geq 1-\exp ( - \beta^{9/4}n/2^m )$.
The result follows.
\end{proof}

\subsection{Sparsification}\label{1-9481-4098214-09}

Let us start with the proof of Lemma~\ref{lem:cor-3}.

\medskip

\begin{proof}[Proof of Lemma~\ref{lem:cor-3}]
Fix a set $\cJ \subset \{-1,1\}^m$ of cardinality $k$. 
We can view $|\cI \cap \cJ|$ as a sum of $\omega$ indicators sampled without replacement from a population of $k$ ones and $2^m - k$ zeros. Therefore, Bernstein's inequality for sampling without replacement (Lemma~\ref{lem:bern-wo}) yields that
$$
\p \Big\{ \Big| |\cI \cap \cJ| - \frac{\omega k}{2^m} \Big| > \frac{\omega k'}{2^m} \Big\} \le 2 \exp \Big( \frac{ - 3 (\omega k' / 2^m)^2  }{ 6 \sigma^2 \omega + 2 \omega k' / 2^m } \Big) . 
$$
for any $k' \ge k$ and 
$$
\sigma^2 := \frac{1}{2^m} \Big[ (2^m - k) \Big( \frac{ k }{ 2^{m} } \Big)^2 + k \Big( 1 - \frac{ k }{ 2^{m} } \Big)^2 \Big] \le \frac{ 2 k }{ 2^{m} } . 
$$
In particular, it follows that 
$$
\p \Big\{ |\cI \cap \cJ| > \frac{2 \omega k'}{2^m} \Big\} \le 2 \exp \Big( \frac{ - 3 \omega k'  }{ 14 \cdot 2^m } \Big) . 
$$

Now we set $\cJ = \big\{ s \in \{-1,1\}^m : |R_s \cap R'_s| \le \beta^2 n / 2^{m+6} \big\}$. By conditioning on \eqref{eq:cor-bd}, we have $k \le k' := C \, 2^m m^{3/2} \, \beta^{1/4} \log(1/\beta)$. 
In view of the condition
$$
\omega m^{3/2} \, \beta^{1/4} \log(1/\beta) \gg (\log n)^{1+\delta} , 
$$
the result follows immediately. 
\end{proof}

The sparsification step of the algorithm hinges on the following lemma. 


\begin{lemma}[Sparsification]
\label{lem:spa}
Fix a constant $\gamma\in(0,1)$ and an even integer $k\in \N_+$. Let $\Omega$ and $\Omega'$ be two finite sets, and let
$$
\Omega=\bigcup_{i=1}^k\Omega_i\quad\mbox{and}\quad \Omega'=\bigcup_{i=1}^k\Omega_i'
$$
be partitions of $\Omega$ and $\Omega'$ respectively, such that 
$$\gamma |\Omega'|/k \le |\Omega_i'| \le \gamma^{-1} |\Omega'|/k$$
for all $i \in [k]$.
Furthermore, let $w \in \{2, 3 , \dots, k/2\}$ and let $I$ be a uniform random subset of $[k]$ of cardinality $2w$. Then,
for any $L\geq 1$ and $\rho\in (0, 1/4)$ such that $\rho w$ is an integer, we have
$$
\Prob\Big\{\big|\big\{i\in I:\;\exists \, j \in I \setminus \{ i \} \, \text{ s.t. } |\Omega_i\cap \Omega_j'|\geq L |\Omega'|/k^2\big\}\big|\geq 2 \rho w\Big\}\leq
\bigg(\frac{8w^3}{\gamma L}\bigg)^{ \rho w}.
$$
\end{lemma}

\begin{proof}
Let us define 
\begin{align*}
\Event := \Big\{\big|\big\{i\in I:\;\exists \, j \in I \setminus \{ i \} \, \text{ s.t. } |\Omega_i\cap \Omega_j'|\geq L |\Omega'|/k^2\big\}\big|\geq 2 \rho w\Big\} 
\end{align*}
which is the event we aim to control. 
Let $(i_1,i_2,\dots,i_{2w})$ be a uniform random permutation
of indices in $I$; in other words, $i_1,i_2,\dots,i_{2w}$ are indices chosen from $[k]$ uniformly at random without replacement. 
We define another event 
\begin{align*}
\Event' := \Big\{\forall \,
v \in [  \rho w ]  , \, \exists \, u \in [2v-1] \, \text{ s.t. } |\Omega_{i_{2v}}\cap \Omega_{i_{u}}'|\geq L |\Omega'|/k^2\Big\} . 
\end{align*}
We claim that
\begin{align}
\Big(\frac{1}{2w}\Big)^{2 \rho w} \cdot \Prob (\Event)
\le  \p ( \Event' ) 
\le \Big(\frac{2w}{\gamma L}\Big)^{ \rho w} ,
\label{eq:claim-2}
\end{align}
from which the conclusion follows immediately. 

To prove the first inequality in \eqref{eq:claim-2}, fix any realization $I_0$ of $I$ such that the event $\Event$ holds. 
Then it suffices to show that 
$$\Prob (\Event' \mid I=I_0 )
\geq \Big(\frac{1}{2w}\Big)^{2 \rho w},$$
where the randomness is with respect to the uniform random permutation $(i_1,i_2,\dots,i_{2w})$ of indices in $I_0$. 
To prove this lower bound, fix $v \in [  \rho w ]$ 
and condition further on any realization of the indices  $i_1,i_2,\dots,i_{2v-2}$. 
By the definition of $\Event$, there exist
$i \in I_0 \setminus \{ i_1,i_2,\dots,i_{2v-2} \}$ and $j \in I_0 \setminus \{ i \}$ such that $|\Omega_i\cap \Omega_j'|\geq L |\Omega'|/k^2$.
First, assume that $j\notin \{ i_1,i_2,\dots,i_{2v-2} \}$. Then, since $(i_{2v-1},i_{2v},\dots,i_{2w})$ is conditional uniform
permutation of $I_0 \setminus \{ i_1,i_2,\dots,i_{2v-2} \}$, we would have $i_{2v-1}=j$ and $i_{2v}=i$
with conditional probability at least $(\frac{1}{2w})^2$, that is,
\begin{align*}
\Prob_{i_{2v-1}, i_{2v}} \Big\{|\Omega_{i_{2v}}\cap \Omega_{i_{2v-1}}'|\geq L |\Omega'|/k^2 \;\big|\; I=I_0; i_1,i_2,\dots,i_{2v-2}\Big\} \ge \Big( \frac{1}{2w} \Big)^2.
\end{align*}
Otherwise, if $j\in \{ i_1,i_2,\dots,i_{2v-2} \}$ then, since $i_{2v}=i$ with conditional probability at least $\frac{1}{2w}$,
we would have
\begin{align*}
\Prob_{i_{2v-1}, i_{2v}} \Big\{ \exists \, u \in [2v-2] \, \text{ s.t. }  |\Omega_{i_{2v}}\cap \Omega_{i_{u}}'|\geq L |\Omega'|/k^2 \;\big|\; I=I_0; i_1,i_2,\dots,i_{2v-2}\Big\} \ge \frac{1}{2w}.
\end{align*}
Combining the two cases, we get that regardless of the location of $j$,
\begin{align*}
\Prob_{i_{2v-1}, i_{2v}} \Big\{ \exists \, u \in [2v-1] \, \text{ s.t. }  |\Omega_{i_{2v}}\cap \Omega_{i_{u}}'|\geq L |\Omega'|/k^2 \;\big|\; I=I_0; i_1,i_2,\dots,i_{2v-2}\Big\} \ge \Big( \frac{1}{2w} \Big)^2 ,
\end{align*}
where the randomness is with respect to $i_{2v-1}$ and $i_{2v}$. 
In view of the definition of $\Event'$, combining the above bound for all $v = 1, 2, \dots, \rho w$ yields that  $\Prob (\Event' \mid I=I_0 )
\geq \big(\frac{1}{2w}\big)^{2 \rho w}.$ 

Next, we prove the second inequality in \eqref{eq:claim-2}. 
Fix any $v \in [  \rho w ]$ and condition on any realization of the indices  $i_1,i_2,\dots,i_{2v-1}$. Then we have, by a union bound and Markov's inequality (where the randomness is with respect to $i_{2v}$), 
\begin{align*}
&\Prob_{i_{2v}} \Big\{ \exists \, u \in [2v-1] \, \text{ s.t. }  |\Omega_{i_{2v}}\cap \Omega_{i_{u}}'|\geq L |\Omega'|/k^2 \;\big|\;i_1,i_2,\dots,i_{2v-1}\Big\}\\
&\le \sum_{u = 1}^{2v-1}  \frac{ \E \big[ |\Omega_{i_{2v}}\cap \Omega_{i_{u}}'| \;\big|\;i_1,i_2,\dots,i_{2v-1} \big] }{L |\Omega'|/k^2} \\
&= \frac{ 1 }{L |\Omega'|/k^2}   \sum_{u = 1}^{2v-1}  \sum_{j \in [k] \setminus \{i_1, \dots, i_{2v-1}\}} \frac{  |\Omega_{j}\cap \Omega_{i_{u}}'| }{ k - (2v - 1) }  \\
&\le \frac{ 1 }{L |\Omega'|/k^2}   \sum_{u = 1}^{2v-1}  \frac{  | \Omega_{i_{u}}'| }{ k - (2v - 1) }  \\
&\le \frac{ 2v - 1 }{L |\Omega'|/k^2} \cdot \frac{\gamma^{-1}|\Omega'|/k}{k-(2v-1)}
\leq
\frac{2 w}{\gamma L},
\end{align*}
where the last step follows from the facts that $2v - 1 \le 2 (\rho w + 1) - 1 \le w \le k/2$ and that $x k / (k-x) \le 2x$ if $0 < x \le k/2$. 
In view of the definition of $\Event'$, combining the above bound for all $v = 1, 2, \dots, \rho w$ yields that $\p ( \Event' )  \le \big(\frac{2w}{\gamma L}\big)^{ \rho w}$. 
\end{proof}

We are ready to prove Lemma~\ref{lem:spa-2}. 

\medskip

\begin{proof}[Proof of Lemma~\ref{lem:spa-2}]
Recall that $\{R_s\}_{s \in \{-1,1\}^m}$ and $\{R'_s\}_{s \in \{-1,1\}^m}$ are partitions of $\SB$ and $\SB'$ respectively, and that $|\SB| = |\SB'| = \beta n / 2$. We apply Lemma~\ref{lem:spa} with $k = 2^m$, $\gamma = 1/4$, $w = \omega /2$, $\rho$ such that $\rho w = \lfloor (\log n)^{1+\delta} \rfloor$, and $L = 8 e \omega^3$, to obtain that 
$$
\Prob\Big\{\Big|\Big\{ s \in \cI:\;\exists \, t \in \cI \setminus \{s\} \text{ s.t. } |R_s \cap R'_{t}| \geq \frac{ e \, \omega^3 \beta n }{ 2^{2m-2} } \Big\}\Big|\geq 2 (\log n)^{1+\delta} \Big\}\leq
e^{- (\log n)^{1+\delta} } .
$$
Let us define
$$
\cI_1 := \Big\{ s \in \cI:\;\forall \, t \in \cI \setminus \{s\}, \, |R_s \cap R'_{t}| \le \frac{ e \, \omega^3 \beta n }{ 2^{2m-2} } \Big\} , 
$$
Then we have $|\cI \setminus \cI_1| \le 2(\log n)^{1+\delta}$ with probability at least $1 - e^{- (\log n)^{1+\delta} }$, and for any $s \in \cI_1$, 
\begin{align*}
|R_s \cap R'_{\cI \setminus \{s\}}| 
= \sum_{t \in \cI \setminus \{s\}} |R_s \cap R'_t| 
\le \sum_{t \in \cI \setminus \{s\}} \frac{ e \, \omega^3 \beta n }{ 2^{2m-2} }
\le \frac{ e \, \omega^4 \beta n }{ 2^{2m-2} } .
\end{align*}

The same argument works with the roles of $\{R_s\}$ and $\{R'_s\}$ swapped. 
Therefore, 
if we define
$$
\tilde \cI := \Big\{ s \in \cI:\;\forall \, t \in \cI \setminus \{s\}, \, |R_s \cap R'_{t}| \lor |R'_s \cap R_{t}| \le \frac{ e \, \omega^3 \beta n }{ 2^{2m-2} } \Big\} , 
$$
then $|\cI \setminus \tilde \cI| \le 4(\log n)^{1+\delta}$ with probability at least $1 - 2 e^{- (\log n)^{1+\delta} }$ by a union bound. Moreover, for any $s \in \tilde \cI$, 
\begin{align*}
|R_s \cap R'_{\cI \setminus \{s\}}| \lor |R'_s \cap R_{\cI \setminus \{s\}}| 
\le \frac{ e \, \omega^4 \beta n }{ 2^{2m-2} } .
\end{align*}

Finally, for any $i \in \SC \cap \SC'$, the quantity $\eta(i)_s$ defined in \eqref{eq:eta-s} is the deviation of the binomial random variable $|\neigh_G(i; R_s \cap R'_{\cI \setminus \{s\}})|$ from its mean. 
Thus, by Bernstein's inequality, 
\begin{align*}
\p \bigg\{ |\eta(i)_s|  \ge \sqrt{ \frac{ p n }{ 2^m (\log n)^{ \delta} } } \bigg\}  
&\le 2 \exp \bigg( - \min \bigg\{ \frac{ 2^{m - 6} }{  (\log n)^\delta \omega^4 \beta } , \frac{ 3 }{4}  \sqrt{ \frac{ p n }{ 2^m (\log n)^{ \delta} } } \bigg\} \bigg) \\
&\le 2 \exp \big( - (\log n)^{1+\delta} \big) ,
\end{align*}
where the second inequality holds in view of \eqref{eq:def-beta}, \eqref{eq:def-m} and \eqref{eq:def-omega} so that   
$$
\frac{ 2^{m-6} }{  (\log n)^{\delta} \omega^4 \beta }  \ge \frac{ 2^{m} }{  (\log n)^{4 + 9 \delta} } 
\ge (\log n)^{1+\delta} , \quad 
\frac{ p n }{ 2^m (\log n)^{ \delta} } \gg (\log n)^4 . 
$$
The same bound also holds for $\eta'(i)_s$ in place of $\eta(i)_s$. 
Taking a union bound over $i \in \SC \cap \SC'$ and $s \in \{-1,1\}^m$ finishes the proof.
\end{proof}

\subsection{Correlation of vertex signatures}\label{108746104987}

Everywhere in this subsection, we assume that $n\geq n_0$ for a sufficiently large $n_0=n_0(\delta)$ and
assume \eqref{eq:conditions}, \eqref{eq:def-beta}, \eqref{eq:def-m} and \eqref{eq:def-omega}. 
We continue to use the notation $R_{\cJ}$ and $R'_{\cJ}$ defined in \eqref{eq:rj}.

\begin{lemma}[Correlation of signs]
\label{lem:cor-sign}
Condition on a realization of $\SA, \SA', \SB, \SB', \SC, \SC'$, $G(\SA \cup \SB)$, $G'(\SA' \cup \SB')$, and $\cI \subset \{-1,1\}^m$. 
Fix a vertex $i \in \SC \cap \SC'$ and an index $s \in \cI$. 
Furthermore, condition on a realization of $\neigh_G (i; R_s \cap R'_{\cI \setminus \{s\}} )$ and $\neigh_{G'} (i; R'_s \cap R_{\cI \setminus \{s\}} )$. 
Let $\tilde \p$ denote the conditional probability. 

Let $k_1 := |R_s \cap R'_s|$, $k_2 := |R_s \setminus R'_{\cI}|$ and $k'_2 := |R'_s \setminus R_{\cI}|$. 
Moreover, write $\eta \equiv \eta(i)_s$ and $\eta' \equiv \eta'(i)_s$ which are defined in \eqref{eq:eta-s} and \eqref{eq:eta-sp} respectively. 
Then, there exist universal constants $C, c > 0$ such that the following holds. 
If $k_1 p \ge C$ and $|\eta| \lor |\eta'| \le c \sqrt{k_1 p}$, then 
\begin{align}
&\tilde\p \Big\{ \sign \big( |\neigh_{G}(i; R_s)|-p |R_s| \big) = \sign \big( |\neigh_{G'}(i; R'_s)|-p |R'_s| \big) \Big\}  \notag \\
&\ge \frac 12 
+ c \sqrt{ \Big( \frac{k_1}{k_2} \land 1 \Big) \Big( \frac{k_1}{k'_2} \land 1 \Big) }
- 6 \exp \Big(- c \big( \alpha^{-1} \land \sqrt{k_1 p} \, \big) \Big) - 2 \exp \Big( - \frac{k_1 p}3 \Big) . \label{eq:bd-2}
\end{align}
Without assuming $k_1 p \ge C$ or $|\eta| \lor |\eta'| \le c \sqrt{k_1 p}$, we have that for any $t>0$, 
\begin{align}
&\tilde\p \Big\{ \sign \big( |\neigh_{G}(i; R_s)|-p |R_s| \big) = \sign \big( |\neigh_{G'}(i; R'_s)|-p |R'_s| \big) \Big\} \notag \\
&\ge \frac 12 - 2 \exp(-t) - C \Big( \frac{t + \sqrt{t k_1 p} + \eta}{ \sqrt{ k_2 p} } \land 1 \Big) \Big( \frac{t + \sqrt{t k_1 p} + \eta'}{ \sqrt{ k'_2 p} } \land 1 \Big)  . \label{eq:bd-0} 
\end{align}
\end{lemma}
\begin{remark}
The lemma should be interpreted as follows. If the differences
$|\neigh_G(i; R_s \cap R'_{\cI \setminus \{s\}})| - p \cdot |R_s \cap R'_{\cI \setminus \{s\}}|$ and
$|\neigh_{G'}(i; R'_s \cap R_{\cI \setminus \{s\}})| - p \cdot |R'_s \cap R_{\cI \setminus \{s\}}|$ are small then
``a non-negligible part of the randomness'' of $|\neigh_{G}(i; R_s)|-p |R_s|$ and $|\neigh_{G'}(i; R'_s)|-p |R'_s|$ comes from the
values of $|\neigh_{G}(i; R_s\cap R_s')|$ and $|\neigh_{G'}(i; R_s\cap R_s')|$, which are strongly correlated because of the
correlations between the edges of $G$ and $G'$. Therefore, in that regime the signs of $|\neigh_{G}(i; R_s)|-p |R_s|$
and $|\neigh_{G'}(i; R'_s)|-p |R'_s|$ should coincide with (conditional) probability considerably greater than $1/2$.
On the other hand, if no bound on $|\eta| \lor |\eta'|$ is available, then the signs still agree with probability at least $(\frac{1}{2}-\mbox{a small number})$.
\end{remark}

\begin{proof}[Proof of Lemma~\ref{lem:cor-sign}]
In view of the decomposition
$$
R_s = \big( R_s \cap R'_s \big) \cup \big( R_s \setminus R'_{\cI} \big) \cup \big( R_s \cap R'_{\cI \setminus \{s\}} \big)  ,
$$
we have that 
\begin{align}
|\neigh_G(i; R_s)| - p |R_s| 
= \big( |\neigh_G(i; R_s \cap R'_s)| - k_1 p \big) + \big( |\neigh_G(i; R_s \setminus R'_{\cI})| - k_2 p \big) + \eta . 
\label{eq:decomp1}
\end{align}
Analogously, it holds that 
\begin{align}
|\neigh_{G'}(i; R'_s)| - p |R'_s| 
= \big( |\neigh_{G'}(i; R'_s \cap R_s)| - k_1 p \big) + \big( |\neigh_{G'}(i; R'_s \setminus R_{\cI})| - k'_2 p \big) + \eta' . 
\label{eq:decomp2}
\end{align}
Note that $\eta$ and $\eta'$ are constants conditional on $\neigh_G (i; R_s \cap R'_{\cI \setminus \{s\}} )$ and $\neigh_{G'} (i; R'_s \cap R_{\cI \setminus \{s\}} )$. 
Moreover, the neighborhoods of $i$ in $R_s \cap R'_s$, $R_s \setminus R'_{\cI}$ and $R'_s \setminus R_{\cI}$ are independent, allowing us to analyze them separately.  

Applying Lemma~\ref{lem:neigh-diff} with $\Gamma = G$, $\Gamma' = G'$, $J = J' = R_s \cap R'_s$ and $t = c_0 \big( \frac{1}{\alpha} \land \sqrt{ k_1 p } \, \big)$ for a universal constant $c_0 \in (0, 1)$ to be chosen later, we obtain\footnote{For notational simplicity, we will not necessarily substitute $c_0 \big( \frac{1}{\alpha} \land \sqrt{ k_1 p } \, \big)$ for $t$ in the proof.}
\begin{align*}
&\tilde\p ( \cE_1^c ) \le 6 \exp(- t ) + 2\exp( - k_1 p/3 ), \quad \text{where} \\ 
&\cE_1 := \Big\{ \big| |\neigh_G(i; R_s \cap R'_s)| - |\neigh_{G'}(i; R_s \cap R'_s)| \big| \le 8 \sqrt{ c_0 k_1 p }  \Big\}.
\end{align*}
Here, we have used that, under our assumptions on parameters,
$$
4t + 4 \sqrt{ t \alpha p |J \cap J'| } + 4 \sqrt{ t p |J \triangle J'| }
\leq 8 \sqrt{ c_0 k_1 p }.
$$

In addition, since $|\neigh_G (i; R_s \cap R'_s)| \sim \Bin( k_1, p )$, Lemma~\ref{lem:bin-cdf-gap} yields 
\begin{align*}
&\tilde\p ( \cE_2^c ) \le C_1 \sqrt{ c_0 } ,  \quad \text{where} \\
&\cE_2 := \Big\{ \big| |\neigh_G (i; R_s \cap R'_s)| - k_1 p \big| > 16 \sqrt{ c_0 k_1 p }  \Big\} ,
\end{align*}
for a universal constant $C_1>0$. 
The rest of the proof is structured as follows. First, we bound the probability in question conditional on the
event $\cE_1 \cap \cE_2$ and assuming that $|\eta| \lor |\eta'|$ is small, and that $k_1 p$ is bounded below by a large constant.
Next, we provide a lower bound for the probability conditional on the event $\cE_1 \cap \cE_2^c$, again under the assumption
on $|\eta| \lor |\eta'|$ and $k_1p$.
Third, we combine the two bounds to obtain \eqref{eq:bd-2}.
Finally, we get \eqref{eq:bd-0}.

\medskip
\noindent
{\bf On the event $\cE_1 \cap \cE_2$.} 
Let us define events
\begin{align*}
&\cE_+ := \Big\{ \Big( |\neigh_G (i; R_s \cap R'_s)| - k_1 p \Big) \land \Big( |\neigh_{G'} (i; R_s \cap R'_s)| - k_1 p \Big) > 8 \sqrt{ c_0 k_1 p } \Big\} , \\
&\cE_- := \Big\{ \Big( |\neigh_G (i; R_s \cap R'_s)| - k_1 p \Big) \lor \Big( |\neigh_{G'} (i; R_s \cap R'_s)| - k_1 p \Big) < - 8 \sqrt{ c_0 k_1 p } \Big\} . 
\end{align*}
Clearly, we have $\cE_1 \cap \cE_2 \subset \cE_+ \cup \cE_-$. 

Next, take $r_2, r_3 \ge 2$ (to be chosen later). 
Since $|\neigh_G(i; R_s \setminus R'_{\cI})|$ and $|\neigh_{G'}(i; R'_s \setminus R_{\cI})|$ are independent $\Bin(k_2, p)$ and $\Bin(k'_2, p)$ random variables respectively,  Lemma~\ref{lem:bin-lower} gives the estimates
\begin{align*}
&\tilde\p\{ |\neigh_G(i; R_s \setminus R'_{\cI})| - k_2 p + r_2 \ge 0 \} = 1/2 + \delta_2 , \quad
\mbox{ for some }\delta_2 \geq c_1 \Big( \frac{ r_2 }{ \sqrt{k_2 p} } \land 1 \Big) , \\
&\tilde\p\{ |\neigh_{G'}(i; R'_s \setminus R_{\cI})| - k'_2 p + r_3 \ge 0 \} = 1/2 + \delta_3 , \quad
\mbox{ for some }\delta_3\geq c_1 \Big( \frac{ r_3 }{ \sqrt{k'_2 p} } \land 1 \Big) ,
\end{align*}
for a universal constant $c_1>0$. It follows that
\begin{align}
&\tilde\p\Big\{ \sign \big( |\neigh_G(i; R_s \setminus R'_{\cI})| - k_2 p + r_2 \big) = \sign \big( |\neigh_{G'}(i; R'_s \setminus R_{\cI})| - k'_2 p + r_3 \big) \Big\} \notag \\
&= (1/2 + \delta_2) (1/2 + \delta_3) + (1/2 - \delta_2) (1/2 - \delta_3) \notag \\
&= 1/2 + 2 \delta_2 \delta_3 
\ge \frac 12 + 2 c_1^2 \Big( \frac{ r_2 }{ \sqrt{k_2 p} } \land 1 \Big) \Big( \frac{ r_3 }{ \sqrt{k'_2 p} } \land 1 \Big) .  \label{eq:sign-same}
\end{align}
A symmetric argument shows that, if $r_2, r_3 \le -2$, then \eqref{eq:sign-same} holds with $r_2$ and $r_3$ replaced by $-r_2$ and $-r_3$ respectively. 

Finally, in view of the equations \eqref{eq:decomp1} and \eqref{eq:decomp2}, we would like to set 
\begin{align}
r_2 = |\neigh_G(i; R_s \cap R'_s)| - k_1 p + \eta , \quad
r_3 = |\neigh_{G'}(i; R_s \cap R'_s)| - k_1 p + \eta' 
\label{eq:r2r3}
\end{align}
in \eqref{eq:sign-same} to derive a lower bound on the probability that
\begin{align}
\sign \big( |\neigh_G(i; R_s )| - p|R_s| ) = \sign \big( |\neigh_{G'}(i; R'_s)| - p|R'_s| \big)  . 
\label{eq:signeq}
\end{align}
By the assumption that $|\eta| \lor |\eta'| \le c \sqrt{k_1 p}$ for $c>0$ sufficiently small compared to $c_0$, on the event $\cE_+$, we have $r_1 \land r_2 > 4 \sqrt{ c_0 k_1 p }$, while on the event $\cE_-$, we have $r_1 \lor r_2 < -4 \sqrt{ c_0 k_1 p }$. 
Note that $4 \sqrt{ c_0 k_1 p} \ge 2$ by the assumption that $k_1 p$ exceeds a large constant. 
As a result, conditional on {\it any} realization of $\neigh_G(i; R_s \cap R'_s)$ and $\neigh_{G'}(i; R_s \cap R'_s)$ such that
$\cE_+ \cup \cE_-$ holds, we can set $r_2$ and $r_3$ as in \eqref{eq:r2r3} so that \eqref{eq:sign-same} applies to give
that \eqref{eq:signeq} holds with probability at least $\frac12 + c_2 \sqrt{ \big( \frac{ c_0 k_1 }{ k_2 } \land 1 \big) \big( \frac{ c_0 k_1 }{ k'_2 } \land 1 \big) }$
for a universal constant $c_2 \in(0,1/2]$, for which we only use the randomness of $\neigh_G(i; R_s \setminus R'_{\cI})$ and $\neigh_{G'}(i; R'_s \setminus R_{\cI})$. 
Using that $\cE_1 \cap \cE_2 \subset \cE_+ \cup \cE_-$, we obtain
\begin{align*}
&\tilde\p\Big\{ \sign \big( |\neigh_G(i; R_s )| - p|R_s| \big) = \sign \big( |\neigh_{G'}(i; R'_s)| - p|R'_s| \big) \, \Big| \, \cE_1 \cap \cE_2 \Big\} \\
&\ge \frac 12 + c_2 \sqrt{ \Big( \frac{ c_0 k_1 }{ k_2 } \land 1 \Big) \Big( \frac{ c_0 k_1 }{ k'_2 } \land 1 \Big) },
\quad\mbox{ assuming that $|\eta| \lor |\eta'| \le c \sqrt{k_1 p}$ and $k_1 p \ge C$}.
\end{align*}

\medskip
\noindent
{\bf On the event $\cE_1 \cap \cE_2^c$.} 
The analysis here is similar to the previous part, so we only sketch the proof. 
Let us define 
$$
\cE_3 := \Big\{ \Big| |\neigh_G (i; R_s \cap R'_s)| - k_1 p \Big| \lor \Big| |\neigh_{G'} (i; R_s \cap R'_s)| - k_1 p \Big| \le 24 \sqrt{ c_0 k_1 p } \Big\} .
$$
It is clear that $\cE_1 \cap \cE_2^c \subset \cE_3$. 
Again, we fix $r_2, r_3 \in \R$. In this case, Lemma~\ref{lem:bin-cdf-gap} implies
\begin{align*}
&\tilde\p\{ |\neigh_G(i; R_s \setminus R'_{\cI})| - k_2 p + r_2 \ge 0 \} = 1/2 + \delta_2 , \quad
|\delta_2| \le \frac{ C_1 (|r_2| + 1) }{ \sqrt{k_2 p} } \land \frac 12 , \\
&\tilde\p\{ |\neigh_{G'}(i; R'_s \setminus R_{\cI})| - k'_2 p + r_3 \ge 0 \} = 1/2 + \delta_3 , \quad
|\delta_3| \le \frac{ C_1 (|r_3| + 1) }{ \sqrt{k'_2 p} } \land \frac 12 .
\end{align*}
As a result, 
\begin{align}
&\tilde\p\Big\{ \sign \big( |\neigh_G(i; R_s \setminus R'_{\cI})| - k_2 p + r_2 \big) = \sign \big( |\neigh_{G'}(i; R'_s \setminus R_{\cI})| - k'_2 p + r_3 \big) \Big\} \notag \\
&= 1/2 + 2 \delta_2 \delta_3 
\ge \frac 12 -  2 \Big( \frac{ C_1 (|r_2| + 1) }{ \sqrt{k_2 p} } \land \frac 12 \Big) \Big( \frac{ C_1 (|r_3| + 1) }{ \sqrt{k'_2 p} } \land \frac 12 \Big) .  
\label{eq:triv}
\end{align}
Conditioning on the event $\cE_1 \cap \cE_2^c \subset \cE_3$, choosing $r_2$ and $r_3$ as in \eqref{eq:r2r3}, and using the fact that $|\eta| \lor |\eta'| \le c \sqrt{k_1 p}$, we conclude that 
\begin{align}
&\tilde\p\Big\{ \sign \big( |\neigh_G(i; R_s )| - p|R_s| \big) = \sign \big( |\neigh_{G'}(i; R'_s)| - p|R'_s| \big) \, \Big| \, \cE_1 \cap \cE_2^c \Big\} \notag \\
&\ge \frac 12 - C_2 \sqrt{ \Big( \frac{ c_0 k_1 }{ k_2 } \land 1 \Big) \Big( \frac{ c_0 k_1 }{ k'_2 } \land 1 \Big) },
\quad\mbox{ assuming that $|\eta| \lor |\eta'| \le c \sqrt{k_1 p}$ and $k_1 p \ge C$},
\label{eq:part-2}
\end{align}
for a universal constant $C_2 > 0$. 

\medskip
\noindent
{\bf Proving the bound \eqref{eq:bd-2}.} 
It follows from the bounds on $\tilde\p(\cE_1^c)$ and $\tilde\p(\cE_2^c)$ that 
\begin{align*}
&\tilde\p ( \cE_1 \cap \cE_2^c) \le \tilde\p ( \cE_2^c ) \le C_1 \sqrt{c_0} , \\
&\tilde\p ( \cE_1 \cap \cE_2) = \tilde\p ( \cE_1 ) - \tilde\p ( \cE_1 \cap \cE_2^c ) \ge 1- 6 \exp(-t) - 2\exp( - k_1 p/3 ) - \tilde\p ( \cE_1 \cap \cE_2^c ) . 
\end{align*}
Writing $\kappa := \sqrt{ \big( \frac{ c_0 k_1 }{ k_2 } \land 1 \big) \big( \frac{ c_0 k_1 }{ k'_2 } \land 1 \big) }$ for brevity and combining the above two parts of the proof, we obtain
\begin{align*}
&\tilde\p\Big\{ \sign \big( |\neigh_G(i; R_s )| - p|R_s| ) = \sign \big( |\neigh_{G'}(i; R'_s)| - p|R'_s| \big) \Big\} \\
&\ge \tilde\p ( \cE_1 \cap \cE_2) \Big(  \frac 12 + c_2 \kappa \Big) + \tilde\p ( \cE_1 \cap \cE_2^c) \Big( \frac 12 - C_2 \kappa \Big) \\
&\ge \frac 12 
+ c_2 \kappa 
- 3 \exp(-t) -  \exp( - k_1 p/3 ) \\ 
& \quad - c_2 \kappa \Big[ 6 \exp(-t) +  2\exp( - k_1 p/3 ) + \tilde\p ( \cE_1 \cap \cE_2^c ) + \frac{C_2}{c_2} \tilde\p ( \cE_1 \cap \cE_2^c ) \Big] \\
&\ge \frac 12 
+ c_2 \kappa
- 3 \exp \Big(- c_0 \big( \alpha^{-1} \land \sqrt{k_1 p} \, \big) \Big) - \exp( - k_1 p/3 ) \\ 
& \quad - c_2 \kappa  \Big[ 6 \exp \Big(- c_0 \big( \alpha^{-1} \land \sqrt{k_1 p} \, \big) \Big) +  2\exp( - k_1 p/3 ) +  (1+C_2/c_2) C_1 \sqrt{c_0} \Big] \\
&\ge \frac 12 
+ \frac{c_2}{2} \kappa
- 6 \exp \Big(- c_0 \big( \alpha^{-1} \land \sqrt{k_1 p} \, \big) \Big) - 2\exp( - k_1 p/3 ) , 
\end{align*}
where we recall $t = c_0 \big( \frac 1\alpha \land \sqrt{k_1 p} \, \big)$. It then suffices to choose $c_0 > 0$ sufficiently small depending on $C_1, C_2$ and $c_2$, and use the assumption $k_1 p \ge C$ for $C>0$ sufficiently large depending on $c_0$. 

\medskip
\noindent
{\bf Proving the bound \eqref{eq:bd-0}.} 
The argument is again very similar to the above, so we only sketch the proof. 
By Bernstein's inequality, we have that for every $t > 0$,
\begin{align*}
&\tilde\p ( \cE_4^c ) \le 4 \exp( -t) , \quad\mbox{where}\\
&\cE_4 := \Big\{ \Big| |\neigh_G (i; R_s \cap R'_s)| - k_1 p \Big| \lor \Big| |\neigh_{G'} (i; R_s \cap R'_s)| - k_1 p \Big| \le 4t/3 + 2 \sqrt{ t k_1 p } \Big\} .
\end{align*}
For fixed $r_2, r_3 \in \R$, we again have the bound \eqref{eq:triv} as above. 
Conditioning on the event $\cE_4$ and choosing $r_2$ and $r_3$ as in \eqref{eq:r2r3}, we obtain
\begin{align*}
&\tilde\p\Big\{ \sign \big( |\neigh_G(i; R_s )| - p|R_s| ) = \sign \big( |\neigh_{G'}(i; R'_s)| - p|R'_s| \big) \, \Big| \, \cE_4 \Big\} \\
&\ge \frac 12 -  2 \Big( \frac{ C_1 (|r_2| + 1) }{ \sqrt{k_2 p} } \land \frac 12 \Big) \Big( \frac{ C_1 (|r_3| + 1) }{ \sqrt{k'_2 p} } \land \frac 12 \Big) \\
&\ge \frac 12 - C_3 \Big( \frac{t + \sqrt{t k_1 p} + \eta}{ \sqrt{ k_2 p} } \land 1 \Big) \Big( \frac{t + \sqrt{t k_1 p} +\eta'}{ \sqrt{ k'_2 p} } \land 1 \Big)   
\end{align*}
for a universal constant $C_3 > 0$. It follows that 
\begin{align*}
&\tilde\p\Big\{ \sign \big( |\neigh_G(i; R_s )| - p|R_s| ) = \sign \big( |\neigh_{G'}(i; R'_s)| - p|R'_s| \big)  \Big\} \\
&\ge \Big( 1 - 4\exp(-t) \Big) \Big[ \frac 12 - C_3 \Big( \frac{t + \sqrt{t k_1 p} + \eta}{ \sqrt{ k_2 p} } \land 1 \Big) \Big( \frac{t + \sqrt{t k_1 p} +\eta'}{ \sqrt{ k'_2 p} } \land 1 \Big)  \Big] \\
&\ge \frac 12 - 2 \exp(-t) - C_3 \Big( \frac{t + \sqrt{t k_1 p} + \eta}{ \sqrt{ k_2 p} } \land 1 \Big) \Big( \frac{t + \sqrt{t k_1 p} + \eta'}{ \sqrt{ k'_2 p} } \land 1 \Big) ,
\end{align*}
which completes the proof.
\end{proof}

We now prove Lemmas~\ref{lem:cor-sig} and~\ref{lem:uncor-sig}. 

\medskip

\begin{proof}[Proof of Lemma~\ref{lem:cor-sig}]
%
%
Let $\tilde \p$ and $\tilde \E$ denote the conditional probability and expectation respectively.  
For each $s \in \cI$, we define 
\begin{align*}
\delta_s := \tilde\p \big\{ f(i)_s = f'(i)_s \big\} - 1/2 \in [-1/2, 1/2] . 
\end{align*}
It follows that 
$$
\tilde \E \Big[ \sum_{s \in \cI} \1 \big\{ f(i)_s = f'(i)_s \big\} \Big] = \frac{\omega}{2} + \sum_{s \in \cI} \delta_s . 
$$
With all the conditioning, the randomness of the pair of signatures $(f(i)_s, f'(i)_s)$ is only with respect to $\neigh_G (i; R_s \setminus R'_{\cI \setminus \{s\}} )$ and $\neigh_{G'} (i; R'_s \setminus R_{\cI \setminus \{s\}} )$.
At the same time, the sets $\neigh_G (i; R_s \setminus R'_{\cI \setminus \{s\}} )
\cup \neigh_{G'} (i; R'_s \setminus R_{\cI \setminus \{s\}} )$ are disjoint across $s \in \cI$.
 Thus, the pairs $(f(i)_s, f'(i)_s)$ are conditionally independent across $s \in \cI$. Hoeffding's inequality then yields that, for any $c_3 > 0$ to be chosen later, 
\begin{align}
\sum_{s \in \cI} \1 \big\{ f(i)_s = f'(i)_s \big\} \ge \frac{\omega}{2}  + \Big( \sum_{s \in \cI} \delta_s \Big) - c_3 \beta \omega
\label{eq:cor-3}
\end{align}
with probability at least $1 - \exp(- 2 c_3^2 \beta^2 \omega)$. 

To bound each $\delta_s$ from below, we apply Lemma~\ref{lem:cor-sign}. 
Let $C, c > 0$ be the constants in Lemma~\ref{lem:cor-sign}. We consider three separate cases in the sequel. 

\medskip
\noindent
{\bf Case 1: $s \in \cI \setminus \tilde \cI$.}
By \eqref{eq:c-bd}, we have $|\cI \setminus \tilde \cI| \le 4 (\log n)^{1+\delta}$, so $\sum_{s \in \cI \setminus \tilde \cI} \delta_s \ge -2 (\log n)^{1+\delta}$. 

\medskip
\noindent
{\bf Case 2: $s \in \tilde \cI$, $|R_s \cap R'_s| > \beta^2 n / 2^{m+6}$.}
Let $k_1 := |R_s \cap R'_s|$.  
To apply \eqref{eq:bd-2}, we check that $k_1 p \ge C$ and $|\eta(i)_s| \lor |\eta'(i)_s| \le c \sqrt{k_1 p}$ by virtue of \eqref{eq:def-beta} and \eqref{eq:eta-bd}. 
Moreover, we have $|R_s|, |R'_s| \le \beta n / 2^{m-1}$ by \eqref{eq:card-bd}. Therefore, \eqref{eq:bd-2} gives
\begin{align*}
\delta_s &\ge c \sqrt{ \Big( \frac{ |R_s \cap R'_s|}{ |R_s \setminus R'_{\cI}|} \land 1 \Big) \Big( \frac{ |R_s \cap R'_s|}{ |R'_s \setminus R_{\cI}|} \land 1 \Big) } \\
& \qquad - 6 \exp \Big(- c \big( \alpha^{-1} \land \sqrt{ p |R_s \cap R'_s| } \, \big) \Big) - 2 \exp \Big( - \frac{ p |R_s \cap R'_s| }3 \Big) \\
&\ge c  \frac{ \beta^2 n / 2^{m+6} }{ \beta n / 2^{m-1} } 
- 6 \exp \Big(- c \big( \alpha^{-1} \land \sqrt{p \beta^2 n / 2^{m+6} } \, \big) \Big) - 2 \exp \Big( - \frac{p \beta^2 n / 2^{m+6} }3 \Big) \\
&\ge c_4 \beta ,
\end{align*}
in view of our assumptions on $n, p, \alpha, \beta$ and $m$. 

\medskip
\noindent
{\bf Case 3: $s \in \tilde \cI$, $|R_s \cap R'_s| \le \beta^2 n / 2^{m+6}$.}
By \eqref{eq:card-bd} again, we have 
$$
|R_s \setminus R'_{\cI}| 
\ge |R_s \setminus \SB'| 
= |R_s| - |R_s \cap \SB'| 
\ge \frac{\beta n}{2^{m+3}} - \frac{ \beta^2 n }{2^{m-1} } 
\ge \frac{\beta n}{2^{m+4}} ,
$$
and similarly $|R'_s \setminus R_{\cI}| \ge \frac{\beta n}{2^{m+4}}$. 
Then it follows from \eqref{eq:bd-0} with $t = \log(1/\beta)$ that
\begin{align*}
\delta_s &\ge - 2 \beta - C \Big( \frac{\log(1/\beta) + \sqrt{ p \log(1/\beta) |R_s \cap R'_s| } + \eta(i)_s }{ \sqrt{ p |R_s \setminus R'_{\cI}| } } \land 1 \Big) \\
&\qquad \qquad \qquad \cdot \Big( \frac{\log(1/\beta) + \sqrt{p \log(1/\beta) |R_s \cap R'_s| } + \eta'(i)_s }{ \sqrt{ p |R'_s \setminus R_{\cI}| } } \land 1 \Big) \\
&\ge - 2 \beta - C_2 \frac{\big(\log(1/\beta) + \sqrt{ p \log(1/\beta) \beta^2 n / 2^m } + c_1 \sqrt{ \beta^2 p n / 2^m } \big)^2 }{ p \beta n / 2^{m} } \\
&\ge - C_3 \beta \log(1/\beta),
\end{align*}
by plugging the above bounds and using our assumptions on the parameters. 

\medskip

Combining the above three cases with \eqref{eq:cor-bd-3} yields
\begin{align*}
\sum_{s \in \cI} \delta_s 
&\ge \sum_{s \in \cI \setminus \tilde \cI} \delta_s + \sum_{s \in \tilde \cI} c_4 \beta \cdot \1 \Big\{ |R_s \cap R'_s| > \beta^2 n / 2^{m+6} \Big\} \\
& \quad + \sum_{s \in \tilde \cI} \big[ - C_3 \beta \log(1/\beta) \big] \cdot \1 \Big\{ |R_s \cap R'_s| \leq \beta^2 n / 2^{m+6} \Big\} \\
&\ge -2(\log n)^{1+\delta} + c_4 \beta \cdot \big[ |\tilde \cI| - C \omega m^{3/2} \, \beta^{1/4} \log(1/\beta) \big] \\
& \quad - C_3 \beta \log(1/\beta) \cdot C \omega m^{3/2} \, \beta^{1/4} \log(1/\beta) \\
&\ge c_5 \beta \omega ,
\end{align*}
where the last step follows from that $(\log n)^{1+\delta} \ll \beta \omega$ and $\omega m^{3/2} \, \beta^{1/4} ( \log \frac{1}{\beta} )^2\ll \omega/2 \le |\tilde \cI|$ which hold in view of \eqref{eq:def-beta}, \eqref{eq:def-m} and \eqref{eq:def-omega}. 
Plugging this bound into \eqref{eq:cor-3}, we see that it suffices to take $c_3>0$ sufficiently small depending on $c_5$ to finish the proof. 
\end{proof}

\begin{proof}[Proof of Lemma~\ref{lem:uncor-sig}]
Let $\tilde \p$ and $\tilde \E$ denote the conditional probability and expectation respectively.  
Note that, conditional on \eqref{eq:neigh-cond}, the randomness of the pair of signatures $(f(i)_s, f'(j)_s)$ is only with respect to $\neigh_G (i; R_s \setminus R'_{\cI \setminus \{s\}} )$ and $\neigh_{G'} (j; R'_s \setminus R_{\cI \setminus \{s\}} )$. 
Therefore, $f(i)_s$ and $f'(j)_s$ are conditionally independent for each fixed $s \in \cI$; moreover, the pairs $(f(i)_s, f'(j)_s)$ are conditionally independent across $s \in \cI$. 

By virtue of the decomposition $R_s = \big( R_s \setminus R'_{\cI \setminus \{s\}} \big) \cup \big( R_s \cap R'_{\cI \setminus \{s\}} \big)$, we have 
$$
|\neigh_{G}(i; R_s)|-p |R_s| 
= |\neigh_G (i; R_s \setminus R'_{\cI \setminus \{s\}} )| - p |R_s \setminus R'_{\cI \setminus \{s\}}| + \eta(i)_s 
$$
where $\eta(i)_s$ is defined in \eqref{eq:eta-s}. 
Moreover, since $|\neigh_G (i; R_s \setminus R'_{\cI \setminus \{s\}} )|$ is a $\Bin(|R_s \setminus R'_{\cI \setminus \{s\}}|, p)$ random variable, it follows from the above equation and Lemma~\ref{lem:bin-cdf-gap} that
\begin{align*}
\Big| \tilde \p \Big\{ f(i)_s = 1 \Big\} - \frac 12 \Big| 
&= \Big| \tilde \p \Big\{ |\neigh_{G}(i; R_s)|-p |R_s| \ge 0  \Big\} - \frac 12 \Big| \\
&= \Big| \tilde \p \Big\{ |\neigh_G (i; R_s \setminus R'_{\cI \setminus \{s\}} )| \ge p |R_s \setminus R'_{\cI \setminus \{s\}}| - \eta(i)_s  \Big\}  - \frac 12 \Big| \\
&\le \frac{ C | \eta(i)_s | }{ \sqrt{ p |R_s \setminus R'_{\cI \setminus \{s\}}| } }  
\le \frac{ C | \eta(i)_s | }{ \sqrt{ p (|R_s| - |R_s \cap \SB'| ) } } .
\end{align*}
Furthermore, applying \eqref{eq:eta-bd}, \eqref{eq:card-bd}, and assumptions on the parameters, we get 
$$
\Big| \tilde \p \Big\{ f(i)_s = 1 \Big\} - \frac 12 \Big|  
\le \frac{ C \sqrt{ pn / (2^m (\log n)^{\delta} ) } }{ \sqrt{ p ( \beta n / 2^{m+3} - \beta^2 n / 2^{m-1} ) } } 
\le (\log n)^{-\delta/3} . 
$$


Next, by the conditional independence of $f(i)_s$ and $f'(j)_s$, we obtain 
$$
\tilde\Exp \Big[ \sum_{s \in \cI} \1 \big\{ f(i)_s = f'(j)_s \big\}  \Big] 
\le  \frac{\omega}{2} \Big( 1 + 2 (\log n)^{-\delta/3} \Big) .
$$
Then, using the conditional independence of $(f(i)_s, f'(j)_s)$ across $s \in \cI$, we can apply Hoeffding's inequality to see that 
$$
\sum_{s \in \cI} \1 \big\{ f(i)_s = f'(j)_s \big\}
\le \frac{\omega}{2} \Big( 1 + 2 (\log n)^{-\delta/3} \Big) + \frac{\beta \omega }{4 \log \log n} 
$$
with probability at least $1 - \exp \big( \frac{ - \beta^2 \omega}{8 (\log \log n)^{2} } \big)$.
This completes the proof in view of \eqref{eq:def-beta}. 
\end{proof}

\section{Probability tools}

\subsection{Concentration inequalities}

\begin{lemma}[Bernstein's inequality]\label{lem: Bern ineq}
Let $X_1, \dots, X_n$ be independent random variables taking values in $[a,b]$. 
Define $\sigma^2 := \frac 1n \sum_{i=1}^n \E \big[ (X_i - \E[X_i] )^2 \big] .$ 
Then, for any $t>0$, we have
$$
\p \Big\{ \Big| \sum_{i=1}^n \big( X_i - \E [ X_i ] \big) \Big| \ge t \Big\}  \le 2 \exp \Big( \frac{ -t^2/2 }{ \sigma^2 n + (b-a) t / 3 } \Big) .
$$
\end{lemma}

\begin{lemma}[Bernstein's inequality for sampling without replacement] 
\label{lem:bern-wo}
Let $\{x_1, \dots, x_N\}$ be a population of $N$ real numbers. For $n \in [N]$, let $X_1, \dots, X_n$ be a uniformly random sample from $\{x_1, \dots, x_N\}$ without replacement. Define
$$
a := \min_{i \in [N]} x_i, \quad
b := \max_{i \in [N]} x_i, \quad
\bar x := \frac 1N \sum_{i=1}^N x_i, \quad 
\sigma^2 := \frac 1N \sum_{i=1}^N (x_i - \bar x)^2 .
$$
Then, for any $t>0$, we have
$$
\p \Big\{ \Big| \sum_{i=1}^n X_i - n \bar x \Big| \ge t \Big\} \le 2 \exp \Big( \frac{ - t^2/2 }{ \sigma^2 n + (b-a) t/3 } \Big) . 
$$
\end{lemma}

\subsection{Binomial distributions}

%

Recall that $F_{k,p}$ denotes the cumulative distribution function of $\Bin(k,p)$. The following lemmas are elementary.

\begin{lemma}\label{l: comp of binom}
For any $p\in(0,1)$, and $a, b, t \in \N_0$ such that $t \le a \le b$, we have 
$$
F_{b,p}(t)\leq F_{a,p}(t)\leq F_{b,p}(t+(b-a)).
$$
\end{lemma}


\begin{lemma} \label{lem:bin-cdf-gap}
Let $X$ be a $\Bin(k, p)$ random variable where $k \in \N_0$ and $p \in (0,1)$. There exists a universal constant $C>0$ such that, for any $t, r > 0$, 
$$
\p \{ X \in [t, t+r] \} \le \frac{C r}{ \sqrt{ k p (1-p) } } .
$$
\end{lemma}


\begin{lemma}
\label{lem:bin-lower}
Let $X$ be a $\Bin(k, p)$ random variable where $k \in \N_0$ and $p \in (0,1)$. 
There exist universal constants $C, c>0$ such that the following holds: If $k p (1-p) \ge C$, then for any $r \ge 2$, we have
$$
\p \{ X > kp - r \} \ge \frac 12 + c \Big( \frac{ r}{ \sqrt{ k p (1-p) } } \land 1\Big) , \quad
\p \{ X < kp + r \} \ge \frac 12 +  c \Big( \frac{ r}{ \sqrt{ k p (1-p) } } \land 1\Big) .
$$
\end{lemma}

\begin{proof}
%
%
We focus on proving the first inequality, as the second is analogous. 
Note that we may assume $r \le \sqrt{ k p (1-p) }$ without loss of generality. 
Let $C'>0$ denote a sufficiently large constant throughout the proof. 

Since the median of a $\Bin(k,p)$ random variable is either $\lfloor kp \rfloor$ or $\lceil kp \rceil$, we have
$$
\p \{ X > \lceil kp \rceil  \} \le 1/2  .
$$
Moreover, it holds that 
$$
\p\{ X = \lfloor kp \rfloor \text{ or } \lceil kp \rceil \} \le C' / \sqrt{ kp(1-p) } \le 0.1 , $$
if $kp(1-p) \ge C \ge 100 (C')^2$. 
In addition, Bernstein's inequality gives 
$$
\p \Big\{ X \in \Big( kp - C' \sqrt{k p (1-p)} , kp + C' \sqrt{k p (1-p)} \Big) \Big\} \ge 0.9 .
$$
Combining these bounds, we easily obtain
$$
\p \Big\{ X \in \Big( kp - C' \sqrt{k p (1-p)} , \lfloor kp \rfloor \Big) \Big\} \ge 0.3 .
$$
Since the binomial p.m.f.\ is monotonically increasing on $\big( kp - C' \sqrt{k p (1-p)} , \lfloor kp \rfloor \big)$, we get 
$$
\p \Big\{ X \in \Big( kp - r , \lfloor kp \rfloor \Big) \Big\} 
\ge \frac{ 0.3 (r - 1) }{ C' \sqrt{k p (1-p)} } .
$$
Together with $\p \{ X \ge \lfloor kp \rfloor \} \ge 1/2 $, this completes the proof.
\end{proof}